\documentclass[a4paper,10pt]{article}
\usepackage{subfig}
\usepackage{amsmath}
\usepackage{amsfonts}
\usepackage{amssymb}
\usepackage{amsthm}
\usepackage{pgf}
\usepackage{tikz}

\usepackage{graphicx}
\usetikzlibrary{arrows,automata}

\newcommand{\opX}{\mathsf{X}}
\newcommand{\opXhat}{\widehat{\opX}}
\newcommand{\bldm}{\mathbf{m}}
\newcommand{\DD}{\mathfrak{D}}
\newcommand{\blde}{\mathbf{e}}
\newcommand{\bldn}{\mathbf{n}}
\newcommand{\capac}{\mathrm{cap}}
\newcommand{\F}{\mathcal{F}}
\newcommand{\N}{\mathcal{N}}
\renewcommand{\H}{\mathcal{H}}
\newcommand{\V}{\mathcal{V}}
\newcommand{\tr}[1]{#1^{\mathtt{t}}}
\newcommand{\bldone}{\mathbf{1}}
\newcommand{\bldzero}{\mathbf{0}}

\newcommand{\bldx}{\mbox{\bf{x}}}
\newcommand{\bldy}{\mathbf{y}}
\newcommand{\bldz}{\mathbf{z}}
\newcommand{\bldr}{\mathbf{r}}
\newcommand{\bldR}{\mathbf{R}}

\newcommand{\bldl}{\mathbf{l}}
\newcommand{\bldf}{\mathbf{f}}
\newcommand{\bldg}{\mathbf{g}}
\newcommand{\bldb}{\mathbf{b}}
\newcommand{\bldv}{\mathbf{v}}
\newcommand{\bldj}{\mathbf{j}}
\newcommand{\bldV}{\mathbf{V}}
\newcommand{\bldu}{\mathbf{u}}
\newcommand{\bldt}{\mathbf{t}}
\newcommand{\bldw}{\mathbf{w}}
\newcommand{\barGG}{{\bar{\G}}}
\newcommand{\RR}{\mathbb{R}}
\newcommand{\ZZ}{\mathbb{Z}}
\newcommand{\G}{\mathcal{G}}
\newcommand{\D}{\mathcal{D}}
\newcommand{\A}{\mathcal{A}}
\newcommand{\B}{\mathcal{B}}
\newcommand{\I}{\mathcal{I}}
\newcommand{\J}{\mathcal{J}}
\newcommand{\ELL}{\mathcal{L}}
\newcommand{\E}{\mathcal{E}}

\newcommand{\W}{\mathcal{W}}
\renewcommand{\S}{\mathcal{S}}

\newcommand{\RLL}{\mathrm{RLL}}
\newcommand{\NAK}{\mathrm{NAK}}
\newcommand{\RWIM}{\mathrm{RWIM}}
\newcommand{\CHARGE}{\mathrm{CHG}}
\newcommand{\EVEN}{\mathrm{EVEN}}
\newcommand{\ODD}{\mathrm{ODD}}
\newcommand{\R}{\mathcal{R}}
\newcommand{\osqr}{^{{\otimes}2}}
\renewcommand{\P}{\mathcal{P}}

\newtheorem{theorem}{Theorem}
\newtheorem{lemma}{Lemma}
\newtheorem{corollary}{Corollary}
\newtheorem{proposition}{Proposition}

\theoremstyle{remark}
\newtheorem{remark}{Remark}

\newtheoremstyle{nonum}{}{}{\upshape}{}{\itshape}{.}{ }{#1 #3}
\theoremstyle{nonum}
\newtheorem{remark*}{Remark}

\makeatletter
\@addtoreset{remark}{theorem}
\@addtoreset{remark}{proposition}
\@addtoreset{remark}{lemma}
\makeatother

\title{Improved Lower Bounds on Capacities of Symmetric $2$-Dimensional Constraints using Rayleigh Quotients}
\author{Erez Louidor and Brian Marcus}

\begin{document}
\maketitle

\begin{abstract}
A method for computing lower bounds on capacities of $2$-dimensional constraints
having a symmetric presentation in either the horizontal or the
vertical direction is presented. The method is a generalization of
the method of Calkin and Wilf (\emph{SIAM J.\ Discrete Math.},
1998). Previous best lower bounds on capacities of certain
constraints are improved using the method. It is also shown how this
method, as well as their method for computing upper bounds on the
capacity, can be applied to constraints which are not of
finite-type. Additionally, capacities of $2$ families of
multi-dimensional constraints are given exactly.
\end{abstract}
\begin{center}
{\bf Index-terms:} Channel capacity, constrained-coding, min-max principle.
\end{center}

\section{Introduction.}
\label{sec:introduction}
Fix an alphabet $\Sigma$ and let $\G$ be a directed graph whose edges are labeled with symbols in $\Sigma$.
Each path in $\G$ corresponds to a finite word obtained by reading the labels of the edges of the path in sequence.
The path is said to generate the corresponding word, and the set of words generated by all
finite paths in the graph
is called a {\em $1$-dimensional constrained system} or a {\em $1$-dimensional constraint}.
Such a graph is called a
{\em presentation} of the constraint. We say that a word satisfies the constraint if it belongs
to the constrained system. 
One-dimensional constraints have found widespread applications in digital storage
systems, where they are used to model the set of sequences that can be written reliably to a medium.
 A central example is the binary runlength-limited constraint, denoted $\RLL(d,k)$ for nonnegative
 integers $0{\leq}d{\leq}k$, consisting of all binary sequences in which the number of `$0$'s between
 conecutive `$1$'s is at least $d$, and each runlength of `$0$'s has length at most $k$. Another
 $1$-dimensional constraint, often used in practice, is the bounded-charge constraint,
  denoted $\CHARGE(b)$, for some positive integer $b$; it consists of all
 words $w_1 w_2{\ldots}w_\ell$, where $\ell{=}0,1,2,{\ldots}$
 and each $w_i$ is either $+1$ or $-1$, such that for all $1{\leq}i{\leq}j{\leq}\ell$,
 $|\sum_{k{=}i}^{j}w_k|{\leq}b$. Other examples of $1$-dimensional constraints are the $\EVEN$
 and $\ODD$ constraints, which contain all finite binary sequences in which the number of `$0$'s between consecutive `$1$'s
 is even and odd, respectively. Presentations for these constraints are given in Figure~\ref{fig:1d-presentations}.

\begin{figure}[ht]
\label{fig:1d-presentations}
\begin{center}
\subfloat[][]{
\label{fig:even-presentation}
\begin{tikzpicture}[->,>=stealth',shorten >=1pt,auto,node distance=2cm,
                    semithick]
  \tikzstyle{every state}=[fill=black!50!white!50,draw=black,text=black,minimum size=7mm]

  \node[state] (A) {};
  \node[state] (B) [right of=A]{};
  \path (A) edge [bend right,below] node {$0$} (B)
    (B) edge [bend right,below] node {$0$} (A)
        (A) edge [loop above] node {$1$} (A);
\end{tikzpicture}
}\qquad
\subfloat[][]{
\label{fig:odd-presentation}
\begin{tikzpicture}[->,>=stealth',shorten >=1pt,auto,node distance=2cm,
                    semithick]
  \tikzstyle{every state}=[fill=black!50!white!50,draw=black,text=black,minimum size=7mm,]
  \node[state] (A) {};
  \node[state] (B) [right of=A]{};
  \path (A) edge [bend right,below] node {$0$} (B)
    (B) edge [bend right,below] node {$0$} (A)
        (B) edge [bend right=45,above] node {$1$} (A);
\end{tikzpicture}
}\\
\subfloat[][]{
\label{fig:chg-presentation}
\begin{tikzpicture}[->,>=stealth',shorten >=1pt,auto,node distance=2cm,
                    semithick]
  \tikzstyle{every state}=[fill=black!50!white!50,draw=black,text=black,minimum size=7mm,]
  \node[state] (A0) {$\displaystyle_{0}$};
  \node[state] (A1) [right of=A0]{$\displaystyle_{1}$};
  \node[state] (A2) [right of=A1]{$\displaystyle_{2}$};
  \node  (DOTS) [right of=A2]{$\ldots$};
  \node[state] (Ab) [right of=DOTS]{$\displaystyle_{b}$};
  \path (A0) edge [bend left,above] node {$+1$} (A1)
    (A1) edge [bend left,above] node {$+1$} (A2)
    (A2) edge [bend left,above] node {$+1$} (DOTS)
    (DOTS) edge [bend left, above] node {$+1$} (Ab)
    (Ab) edge [bend left, below] node {$-1$} (DOTS)
    (DOTS) edge [bend left, below] node {$-1$} (A2)
    (A2) edge [bend left, below] node {$-1$} (A1)
    (A1) edge [bend left, below] node {$-1$} (A0);
\end{tikzpicture}}
\\
\subfloat[][]{
\label{fig:rlldk-presentation}
\begin{tikzpicture}[->,>=stealth',shorten >=1pt,auto,node distance=1.6cm,
                    semithick]
  \tikzstyle{every state}=[fill=black!50!white!50,draw=black,text=black,minimum size=7mm,]
  \node[state] (A0) {$\displaystyle_{0}$};
  \node[state] (A1) [right of=A0]{$\displaystyle_{1}$};
  \node  (DOTS1) [right of=A1]{$\ldots$};
  \node[state] (Ad) [right of=DOTS1]{$\displaystyle_{d}$};
  \node[state, text centered] (Ad1) [right of=Ad]{\makebox[1pt]{$\displaystyle_{d{+}1}$}};
  \node  (DOTS2) [right of=Ad1]{$\ldots$};
  \node[state] (Ak) [right of=DOTS2]{$\displaystyle_{k}$};
  \path (A0) edge [above] node {$0$} (A1)
    (A1) edge [above] node {$0$} (DOTS1)
    (DOTS1) edge [above] node {$0$} (Ad)
    (Ad) edge [above] node {$0$} (Ad1)
    (Ad1) edge [above] node {$0$} (DOTS2)
    (DOTS2) edge [above] node {$0$} (Ak)

        (Ak) edge [bend left=35,above] node {$1$} (A0)
    (Ad1) edge [bend left=34,above] node {$1$} (A0)
    (Ad) edge [bend left,above] node {$1$} (A0)
;
\end{tikzpicture}}

\end{center}
\caption{Presentations of $1$-dimensional constraints \protect\subref{fig:even-presentation} $\EVEN$ ;\protect\subref{fig:odd-presentation} $\ODD$; \protect\subref{fig:chg-presentation} $\CHARGE(b)$;\protect\subref{fig:rlldk-presentation} $\RLL(d,k)$.}
\end{figure}
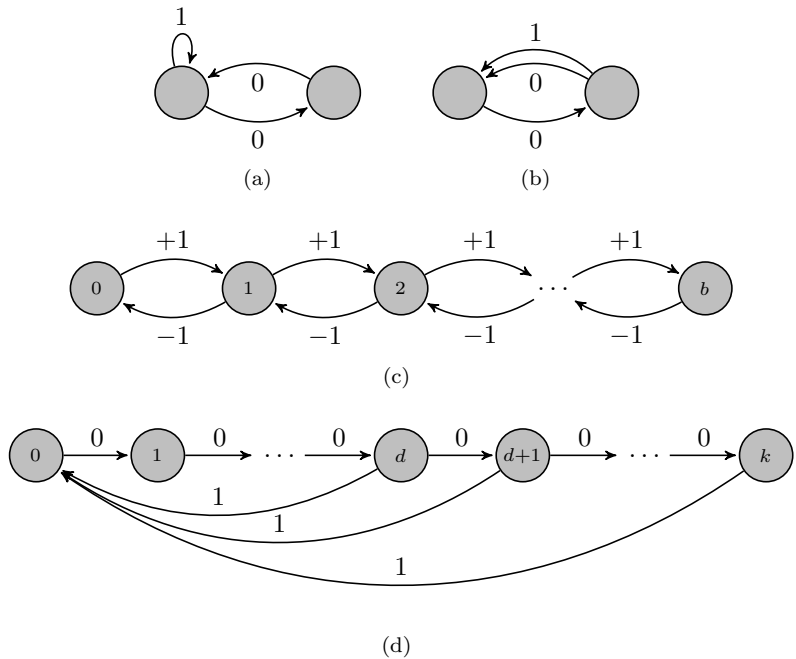

A $1$-dimensional constraint over an alphabet $\Sigma$ is said to have memory $m$,
for some positive integer $m$, if for every word $w$ of more than $m$ letters over $\Sigma$,
in which every subword of $m+1$ consecutive letters satisfies $S$, it holds that $w$ satisfies $S$
as well, and $m$ is the smallest integer for which this holds. A $1$-dimensional constraint with
finite-memory is called a {\em finite-type} constraint. Among our examples, $\RLL(d,k)$ is a finite-type constraint with memory $k$, whereas $\EVEN$, $\ODD$ and $\CHARGE(b)$ for $b{\geq}2$ are not finite-type constraints.

In this work, we consider multidimensional constraints of dimension $\DD$ for some positive integer $\DD$. These are sets of finite-size $\DD$-dimensional arrays with entries over some finite alphabet specified by $\DD$ edge-labeled directed graphs. In Section~\ref{sec:framework} we give a precise definition of what we mean by $\DD$-dimensional constraints. Here we just mention that they are closed under taking subarrays, meaning that if an array belongs to the constraint then any of its $\DD$-dimensional subarrays consisting of ``adjacent'' entries also belongs to the constraint. As before we say that an array satisfies the constraint if it belongs to it. Examples of multidimensional constraints can be obtained by generalizing $1$-dimensional constraints, defining the constraint to consist of all arrays satisfying a given $1$-dimensional constraint $S$ on every ``row'' in every direction along an ``axis'' of the array. We denote such a $\DD$-dimensional constraint by $S^{{\otimes}\DD}$. We will almost exclusively be concerned with $2$ dimensional constraints, namely $\DD=2$. In this case $S\osqr$ is the set of all $2$-dimensional arrays where each row and each column satisfy $S$. A well-known $2$-dimensional constraint studied in statistical mechanics is the so called ``hard-square'' constraint. It consists of all finite-size binary arrays which do not contain $2$ adjacent `$1$'s either horizontally or vertically. Two variations of this constraint are the isolated `$1$'s or ``non-attacking-kings'' constraint, denoted $\NAK$, and the ``read-write-isolated-memory''
constraint, denoted $\RWIM$. The former consists of all finite-size binary arrays in which there are no two adjacent `$1$'s either horizontally, vertically, or diagonally, and the latter consists of all finite-size binary arrays in which there are no two adjacent `$1$'s either horizontally, or diagonally.
Like their $1$-dimensional counterparts, $2$-dimensional constraints play a role in storage systems, where with recent developments, information is written in a true $2$-dimensional fashion rather than using essentially $1$-dimensional tracks.The $\RWIM$ constraint is used to model sequences of states of a binary linear memory in which no two` adjacent entries may contain a `$1$', and in every update, no two adjacent entries are both changed. See~\cite{Cohn} and \cite{Golin-Yong-Zhang-Sheng} for more details.

Let $S$ now be a $\DD$-dimensional constraint over an alphabet $\Sigma$. For a
$\DD$-tuple $\bldm=(m_1,m_2,{\ldots},m_\DD)$ of positive integers,
$S_\bldm$ or $S_{m_1{\times}{\ldots} \times m_\DD}$ denotes the set of all
$m_1{\times}m_2{\times}{\ldots}{\times}m_\DD$ arrays in $S$,
and $\langle\bldm\rangle$ denotes the product of the entries of $\bldm$.
We say that a sequence $\bldm_i =(m_1^{(i)},\ldots,m_\DD^{(i)})$ diverges to infinity,
denoted $\bldm_i\rightarrow\infty$, if $(m_j^{(i)})_{i=1}^{\infty}$ does for each $j$.
The capacity of $S$ is defined by
\begin{equation}
\label{eq:capacity}
\capac(S)=\lim_{i\rightarrow\infty}\frac{\log|S_{\bldm_i}|}{\langle \bldm_i \rangle},
\end{equation}
where $\bldm_i{\rightarrow}\infty$, $|\cdot|$ denotes cardinality, and $\log = \log_2$.
The above limit exists and is independent of the choice of $\left(\bldm_i\right)_{i=1}^{\infty}$,
for any set $S$ of finite-size $\DD$-dimensional arrays over $\Sigma$ which is closed under
taking subarrays. This follows from subadditivity arguments
(see \cite{Kato-Zeger} for a proof for $\DD=2$, which can be generalized to higher dimensions). In fact
\begin{equation}
\label{eq:capac-inf}
\capac(S)=\inf_{\bldm}\frac{\log_2|S_{\bldm}|}{\langle \bldm \rangle}.
\end{equation}
While a closed form
formula for the capacity of $1$-dimensional constraints is known (up to finiding the largest root
of a polynomial), no such formula is known for constraints in higher dimensions, and currently
there are only a few multidimensional constraints for which the capacity is known exactly and
is nonzero (a highly non-trivial example can be found in ~\cite{Kastelyn}).

Let $S$ be a $2$-dimensional constraint over $\Sigma$, and $m$ be a positive integer. The {\em horizontal} (resp. {\em vertical}) {\em strip of height} (resp. {\em width}) {\em $m$} of $S$, denoted $\H_m(S)$ (resp. $\V_m(S)$) is the subset of $S$ given by
$$
\begin{array}{ll}
\displaystyle{\H_m(S)=\bigcup_{n}S_{m{\times}n}} &
\mbox{(resp. }\displaystyle{\V_m(S)=\bigcup_{n}S_{n{\times}m}}\mbox{).}
\end{array}
$$
We show in Section~\ref{sec:framework} that $\H_m(S)$ and $\V_m(S)$ are $1$-dimensional constraints over $\Sigma^m$.

A method for computing very good lower and upper bounds on the capacity of the hard-square
constraint is given in ~\cite{Calkin-Wilf} (see also~\cite{Markley-Paul}).
Their method can be shown to work on $2$-dimensional constraints for which every horizontal
or every vertical strip has memory $1$ and is ``symmetric''; that is, it is closed under
reversing the order of symbols in words. The main contributions of our work are:
\begin{enumerate}
\item We establish a generalization of the method of~\cite{Calkin-Wilf} that gives improved lower bounds on capacities
of $2$-dimensional constraints, for instance for $\NAK$ and $\RWIM$.
\item We show how this generalization as well as the original method for obtaining upper-bounds
may be applied to
a larger class of $2$-dimensional constraints that includes constraints in which the vertical and
 horizontal strips are not necessarily finite-type.
 We illustrate this by computing lower and upper bounds on the capacities of
the $\CHARGE(3)\osqr$ and $\EVEN\osqr$ constraints.
\item We show that $\capac(\CHARGE(2)^{{\otimes}\DD})=2^{-\DD}$ and $\capac(\ODD^{{\otimes}\DD})=1/2$, for all positive integers $\DD$.
\end{enumerate}
Previous work involving applications of the method of~\cite{Calkin-Wilf} and generalizations
include~\cite{Forsch-Justesen},~\cite{Friedland},~\cite{Nagy-Zeger}, and~\cite{Weeks-Blahut}.

\section{Framework.}
\label{sec:framework}
In this section, we define the framework that we use in the rest of the paper.
We deal with a directed graph $G=(V,E)$, sometimes simply called a graph,
with vertices $V$ and edges $E$. For $e{\in}E$ we denote by
 $\sigma_\G(e)$ and $\tau_\G(e)$ the initial and terminal vertices of $e$ in $G$, respectively. We shall omit the subscript $G$ from $\sigma_G$ and $\tau_G$ when the graph is clear from the context. A {\em path of length $\ell$} in $G$ is a sequence of $\ell$ edges $(e_i)_{i=1}^{\ell}{\subseteq}E$, where for $i=1,2,{\ldots},\ell{-}1$, $\tau(e_i)$=$\sigma(e_{i+1})$. The path starts at the vertex $\sigma(e_1)$ and ends at the vertex $\tau(e_\ell)$. A cycle in $G$ is a path that starts and ends at the same vertex
Fix a finite alphabet $\Sigma$. A directed {\em labeled graph $\G$} with labels in
$\Sigma$ is a pair $\G=(G,\ELL)$, where $G=(V,E)$ is a directed graph, and $\ELL:E\rightarrow\Sigma$ is a labeling of the edges of $G$ with symbols of $\Sigma$.
The paths and cycles of $\G$ are inherited from $G$ and we will sometime use $\sigma_\G$ and $\tau_\G$ to denote $\sigma_G$ and $\tau_G$ respectively.

For a labeled graph $\G=((V,E),\ELL)$ with $\ELL:E\rightarrow\Sigma$ and a path $(e_i)_{i=1}^\ell$ of $\G$,
we say the path {\em generates} the word $\ELL(e_1)\ELL(e_2){\ldots}\ELL(e_\ell)$ in $\Sigma^*$. The graph $\G$ is
called {\em lossless} if for any two vertices $u$ and $v$ of $\G$, all paths starting at $u$ and
terminating at $v$ generate distinct words. The graph $\G$ is called {\em deterministic} if
there are no two distinct edges with the same initial vertex and the same label.
Every $1$-dimensional constraint $S$ has a deterministic, and therefore lossless,
presentation~\cite{MRS}.

We introduce two $1$-dimensional constraints defined by general directed graphs. Let
 $G{=}(V,E)$ be a directed graph,
the {\em edge constraint defined by $G$}, denoted $\opX(G)$, is the $1$-dimensional constraint
over the alphabet $E$, presented by $\G=(G,I_E)$ where $I_E$ is the identity map on $E$. For a
graph $G{=}(V,E)$ with no parallel edges, {\em the vertex-constraint defined by $G$}, denoted
$\opXhat(G)$, is the set
$$
\left\{(v_i)_{i=1}^\ell{\subseteq}V{:}
\mbox{$\ell{=}0{,}1{,}2{,}{\ldots}$, and for $1{\leq}i{<}\ell$, $\exists e{\in}E$ s.t.
$\sigma(e){=}v_i$, $\tau(e){=}v_{i+1}$}\right\}.
$$
It's not hard to verify that vertex-constraints and edge-constraints are $1$-dimensional
constraints with memory $1$.  In fact, the vertex constraints are precisely the finite-type
constraints with memory $1$.  And it can be shown that edge constraints are characterized
as follows.  The {\em follower set} of a symbol $a$ in a constraint $S$ is defined to be
$\{b: ab \in S\}$; edge constraints are precisely the constraints with memory 1 such that
any two follower sets are either disjoint or identical~\cite[exercise 2.3.4]{Lind-Marcus}.

We consider multidimensional arrays of dimension $\DD$--a positive integer. For an integer $\ell$ denote by $[\ell]$ the set $\left\{0,1,{\ldots},\ell{-}1\right\}$. For a $\DD$-tuple $\bldm=\left(m_1,\ldots,m_\DD\right)$ of nonnegative integers and a finite set $A$, we call an $m_1\times m_2 \times \ldots \times m_\DD$ $\DD$-dimensional array with entries in $A$, a $\DD$-dimensional array of size $\bldm$ over $A$. We shall index the entries of such an array by $[m_1]{\times}[m_2]{\times}{\ldots}{\times}[m_\DD]$. Let $A^\bldm$ or $A^{m_1{\times}{\ldots}{\times}m_\DD}$ denote the set of all $\DD$-dimensional arrays of size $\bldm$ over $A$. We define $A^{*{\ldots}*}=A^{*^\DD}$, where the number of `$*$'s in the superscript is $\DD$, by
$$
A^{*^\DD}=\bigcup_{\bldm}\Sigma^\bldm,
$$
as the set of all finite-size $\DD$-dimensional arrays with entries in $A$.
Let $\Gamma{\in}A^{*^\DD}$ be such an array. Given an integer $1{\leq}i{\leq}\DD$, {\em a row in direction $i$}
of $\Gamma$ is a sequence of entries of $\Gamma$ of the form
$\left(\Gamma_{(k_1,k_2,\ldots,k_{i-1},j,k_{i+1},\ldots,k_\DD)}\right)_{j=0}^{m_i{-}1}$
for some integers $k_l{\in}[m_l];\;1{\leq}l{\leq}\DD,\,l{\neq}i$. In this paper, for $\DD=2$, we use the convention that direction $1$ is the vertical direction and direction $2$ is the horizontal; thus the columns of a $2$-dimensional array are its rows in direction $1$, and its ``traditional rows'' are its rows in direction $2$.
Let $B$ be a finite set and $\ELL:A\rightarrow B$ be a mapping. We extend $\ELL$ to a mapping $\ELL:A^{*^{\DD}}{\rightarrow}B^{*^{\DD}}$ as follows. For a $\DD$-dimensional array $\Gamma{\in}A^\bldm$, $\ELL(\Gamma)$  is the array in $B^{\bldm}$ obtained by applying $\ELL$ to each entry of $\Gamma$, that is
$$
\left(\ELL(\Gamma)\right)_{\bldj}=\ELL(\left(\Gamma\right)_{\bldj})\;\;,\;\;\bldj{\in}[m_1]{\times}{\ldots}{\times}[m_\DD].
$$
Additionally, for a subset $S{\subseteq}A^{*^\DD}$ we define $\ELL(S)=\{\ELL(\Gamma):\Gamma{\in}S\}$.

We generalize the definition of constrained system to $\DD$ dimensions.
Let $\bar{\G}=(\G_1,\G_2,{\ldots},\G_\DD)$, be a $\DD$-tuple of labeled graphs with the
same set of edges $E$ and the same labeling $\ELL:E\rightarrow\Sigma$. The edge $e$
has $\DD$ pairs of initial and terminal vertices $(\sigma_{\G_i}(e)$, $\tau_{\G_i}(e))$--one
for each graph $\G_i$ in $\barGG$. We say that an array $\Gamma{\in}\Sigma^{*^\DD}$ of size $\bldm$
is {\em generated} by $\barGG$ if there exists an array $\Gamma'{\in}E^{*\DD}$ of size $\bldm$,
such that for $i=1,2,{\ldots},\DD$, every row in direction $i$ of $\Gamma'$ is a path in $\G_i$,
and $\ELL(\Gamma')=\Gamma$.
We call the set of all arrays $\Gamma{\in}\Sigma^{*\DD}$ generated by $\barGG$, the
{\em $\DD$-dimensional constrained system} or the {\em $\DD$-dimensional constraint}
presented by $\barGG$, and denote it by $\opX(\barGG)$. We say that $\barGG$ is a {\em presentation}
for $\opX(\barGG)$.

In~\cite{Halevi-Roth}, $2$-dimensional constrained systems are defined by
vertex-labeled graphs, with a common set of vertices
and a common labelling on the vertices.  It can be shown that their definition (generalized to higher
dimensions) is equivalent to ours. We find it more convenient to use our definition, since, just
as in one dimension, it permits use of parallel edges and often enables a smaller presentation of
a given constraint.

Figure~\ref{fig:2d-presentations} shows presentations for the $\NAK$ and $\RWIM$ constraints defined
in Section~\ref{sec:introduction}. In these presentations $\G_1$ and $\G_2$ describe the vertical
and horizontal constraints on the edges, respectively.
Each edge $\blde=(\blde)_{i,j}$ is a $2{\times}2$ binary matrix of the form
$$
\blde=\left(
\begin{array}{cc}
       (\blde)_{(0,0)} & (\blde)_{(0,1)}\\
       (\blde)_{(1,0)} & (\blde)_{(1,1)}
\end{array}
\right),
$$
and it is labeled by $(\blde)_{(1,1)}$, i.e., the labelling of an edge simply picks out
the entry in the lower-right corner.
For $\NAK$, the edges $E = E_{\NAK}$  are the $2\times 2$ matrices with at most one 1; note that a
rectangular array
of any size satisfies  $\NAK$ iff each of its $2\times 2$ sub-arrays belongs to $E$; similarly,
for $\RWIM$, the edges $E=E_{\RWIM}$ are the elements of $E_{\NAK}$ together with
$$
\left(\begin{array}{cc}
       1 &  0\\
       1 &  0
\end{array}\right) \mbox{ and }
\left(\begin{array}{cc}
       0 &  1\\
       0 &  1
\end{array}\right),
$$
and again
a
rectangular array
of any size satisfies  $\RWIM$ iff each of its $2\times 2$ sub-arrays belongs to $E$.
In the figures, each edge is drawn twice--once in $\G_1$ and once in $\G_2$--and
the matrix identifying it is written next to it. The states are $1 \times 2$ blocks for
$\G_1$ and $2 \times 1$ blocks for $\G_2$; for an edge $\blde$,
$$
\sigma_{\G_1}(e) =
\begin{array}{c}
       (\blde)_{(0,0)} (\blde)_{(0,1)}
\end{array}, \mbox{ and }
\tau_{\G_1}(e) =
\begin{array}{c}
(\blde)_{(1,0)}  (\blde)_{(1,1)}
\end{array},
$$
and
$$
\sigma_{\G_2}(e) =
\begin{array}{c}
       (\blde)_{(0,0)}\\
       (\blde)_{(1,0)}
\end{array}, \mbox{ and }
\tau_{\G_2}(e) =
\begin{array}{c}
       (\blde)_{(0,1)}\\
       (\blde)_{(1,1)}
\end{array}.
$$
It follows that, for both constraints,
and any rectangluar array $\Gamma'{\in}E^{m{\times}n}$ with
each of its rows a path in $\G_2$ and each of its columns a path in $\G_1$, it holds that
$$
\Gamma'_{(i,j)}=\left(\begin{array}{cc}
                     \ELL(\Gamma'_{(i{-}1,j{-}1)})&\ELL(\Gamma'_{(i{-}1,j)})\\
             \ELL(\Gamma'_{(i,j{-}1)})&\ELL(\Gamma'_{(i,j)})
                    \end{array}
        \right).
$$
for $i=1,2,{\ldots},m{-}1$, $j=1,2,{\ldots},n{-}1$. Therefore, the only $2{\times}2$
sub-arrays appearing in the array $\ELL(\Gamma')$ are elements of $E$,
and it follows that $\ELL(\Gamma')$ satisfies the corresponding constraint.
Similarly, it can be shown that any array satisfying the constraint can be
generated by the presentation.

\newcommand{\edge}[4]{\left(\displaystyle^{#1#2}_{#3\mathbf{#4}}\right)}
\begin{figure}[ht]
\label{fig:2d-presentations}
\begin{center}
\subfloat[][]{
\label{fig:NAK-presentation}
\begin{tikzpicture}[->,>=stealth',shorten >=1pt,auto,node distance=1.8cm,
                    semithick]
  \tikzstyle{every state}=[fill=black!5!white!95,draw=black,text=black,minimum size=5mm]
  \node (caption1) at (0,0) {$\G_1{:}$};
  \node[state] (v01) at (0.5,0) {\makebox[1pt]{\small{$01$}}};
  \node[state] (v00) [right of=v01]{\makebox[1pt]{\small{$00$}}};
  \node[state] (v10) [right of=v00]{\makebox[1pt]{\small{$10$}}};

  \node[below=10pt] (caption2) [below of=caption1] {$\G_2:$};
  \node[state,below=10pt] (h01) [below of=v01]{\raisebox{-3pt}[-1pt][-1pt]{$\displaystyle^{0}_{1}$}};
  \node[state]            (h00) [right of=h01]{\raisebox{-3pt}[-1pt][-1pt]{$\displaystyle^{0}_{0}$}};
  \node[state]            (h10) [right of=h00]{\raisebox{-3pt}[-1pt][-1pt]{$\displaystyle^{1}_{0}$}};
  \node        (dummy) [below of=h00]{};

  \path (v01) edge [bend left,above] node {$\edge{0}{1}{0}{0}$} (v00)
    (v00) edge [loop,above]      node {$\edge{0}{0}{0}{0}$} (v00)
    (v00) edge [bend left,below] node {$\edge{0}{0}{0}{1}$} (v01)
    (v00) edge [bend left,above] node {$\edge{0}{0}{1}{0}$} (v10)
    (v10) edge [bend left,below] node {$\edge{1}{0}{0}{0}$} (v00)

        (h01) edge [bend left,above] node {$\edge{0}{0}{1}{0}$} (h00)
    (h00) edge [loop,above]      node {$\edge{0}{0}{0}{0}$} (h00)
    (h00) edge [bend left,below] node {$\edge{0}{0}{0}{1}$} (h01)
    (h00) edge [bend left,above] node {$\edge{0}{1}{0}{0}$} (h10)
    (h10) edge [bend left,below] node {$\edge{1}{0}{0}{0}$} (h00);
\end{tikzpicture}
}\qquad
\subfloat[][]{
\label{fig:RWIM-presentation}
\begin{tikzpicture}[->,>=stealth',shorten >=1pt,auto,node distance=1.8cm,
                    semithick]
  \tikzstyle{every state}=[fill=black!5!white!95,draw=black,text=black,minimum size=5mm]
  \node (caption1) at (0,0) {$\G_1{:}$};
  \node[state] (v01) at (0.5,0) {\makebox[1pt]{\small{$01$}}};
  \node[state] (v00) [right of=v01]{\makebox[1pt]{\small{$00$}}};
  \node[state] (v10) [right of=v00]{\makebox[1pt]{\small{$10$}}};
  \node[below=10pt] (caption2) [below of=caption1] {$\G_2:$};
  \node[state,below=10pt] (h01) [below of=v01]{\raisebox{-3pt}[-1pt][-1pt]{$\displaystyle^{0}_{1}$}};
  \node[state] (h00) [right of=h01]{\raisebox{-3pt}[-1pt][-1pt]{$\displaystyle^{0}_{0}$}};
  \node[state] (h10) [right of=h00]{\raisebox{-3pt}[-1pt][-1pt]{$\displaystyle^{1}_{0}$}};
  \node[state] (h11) [below of=h00]{\raisebox{-3pt}[-1pt][-1pt]{$\displaystyle^{1}_{1}$}};
  \path (v01) edge [bend left,above] node {$\edge{0}{1}{0}{0}$} (v00)
    (v01) edge [loop, above]     node {$\edge{0}{1}{0}{1}$} (v01)
    (v00) edge [loop,above]      node {$\edge{0}{0}{0}{0}$} (v00)
    (v00) edge [bend left,below] node {$\edge{0}{0}{0}{1}$} (v01)
    (v00) edge [bend left,above] node {$\edge{0}{0}{1}{0}$} (v10)
    (v10) edge [bend left,below] node {$\edge{1}{0}{0}{0}$} (v00)
    (v10) edge [loop, above]     node {$\edge{1}{0}{1}{0}$} (v10)

    (h01) edge [bend left,above] node {$\edge{0}{0}{1}{0}$} (h00)
    (h00) edge [loop,above]      node {$\edge{0}{0}{0}{0}$} (h00)
    (h00) edge [bend left,below] node {$\edge{0}{0}{0}{1}$} (h01)
    (h00) edge [bend left,above] node {$\edge{0}{1}{0}{0}$} (h10)
    (h00) edge [bend left]                      (h11)
    (h10) edge [bend left,below] node {$\edge{1}{0}{0}{0}$} (h00)
    (h11) edge [bend left]                      (h00)
node [left=15pt, above=2pt] at (h11) {$\edge{1}{0}{1}{0}$}
node [right=16.5pt, above=2pt] at (h11) {$\edge{0}{1}{0}{1}$};
\end{tikzpicture}
}
\caption{Presentations of $2$-dimensional constraints: \protect\subref{fig:NAK-presentation} $\NAK$ constraint; \protect\subref{fig:RWIM-presentation} $\RWIM$ constraint.}
\end{center}
\end{figure}
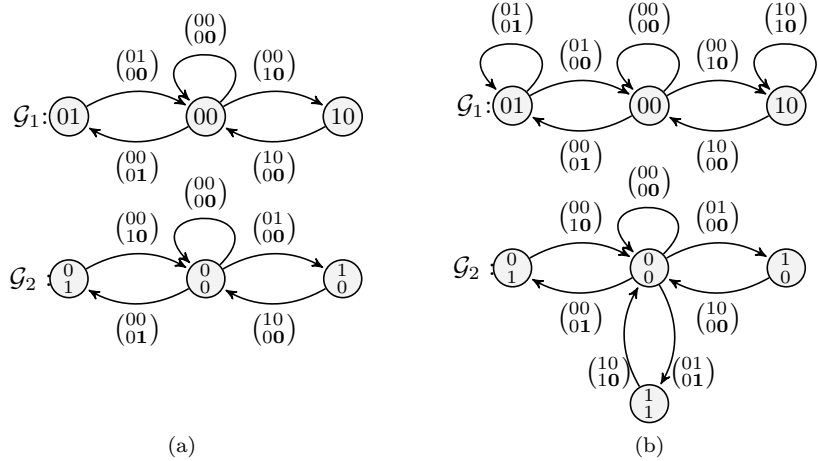
The {\em axial product} of $\DD$ sets $L_1,{\ldots},L_{\DD}{\subseteq}\Sigma^*$,  denoted $L_1{\otimes}L_2{\otimes}{\ldots}{\otimes}L_\DD\subseteq\Sigma^{*^\DD}$, is the set of all arrays $\Gamma{\in}\Sigma^{*^\DD}$ such that for $i=1,2,{\ldots},{\DD}$ every row of $\Gamma$ in direction $i$ belongs to $L_i$. If $L_1{=}L_2{=}{\ldots}{=}L_\DD=L$ we say that the axial-product is {\em isotropic} and denote it by $L^{{\otimes}\DD}$.
Given a presentation $\barGG=((G_1,\ELL),{\ldots},(G_\DD,\ELL))$ for a $\DD$-dimensional constraint
$S$ with a common set of edges $E$, the set
$\opX(G_1){\otimes}{\ldots}{\otimes}\opX(G_\DD)\subseteq E^{*^{\DD}}$ is a
$\DD$-dimensional constraint presented by $((G_1,I_E),{\ldots},(G_\DD,I_E))$,
where $I_E$ is the identity map on $E$.

We say that $\barGG$ is {\em capacity-preserving}
if $\capac(\opX(G_1){\otimes}{\ldots}{\otimes}\opX(G_\DD))=\capac(S)$. For $\DD=1$ any
lossless, and therefore deterministic, presentation of $S$ is capacity-preserving,
since in a lossless graph with $|V|$ vertices there are at most $|V|^2$ paths generating any given
word; for $\DD>1$, the question whether
every $\DD$-dimensional constraint has a capacity-preserving
presentation
is open.
This is a major open problem in symbolic dynamics, although it is usually formulated in a
slightly different manner;
see~\cite{Desai}, where it is shown that for every $\DD$-dimensional constraint $S$ and $\epsilon >0$,
there is a presentation
$\barGG=((G_1,\ELL),{\ldots},(G_\DD,\ELL))$
such that $\capac(S){\le}\capac(\opX(G_1){\otimes}{\ldots}{\otimes}\opX(G_\DD)){<}\capac(S){+}\epsilon$.
We show in the next proposition that the answer to this question is positive, if $S$ is an axial product of $\DD$
one-dimensional constraints.
\begin{proposition}
\label{axial-is-a-constraint}
Let $S_1,S_2,{\ldots},S_{\DD}{\subseteq}{\Sigma}^*$ be $\DD$ one-dimensional
constraints over $\Sigma$, and let
$S=S_1{\otimes}{\ldots}{\otimes}S_{\DD}$, then $S$ is a $\DD$-dimensional constraint over $\Sigma$.
Moreover, $S$ has a capacity-preserving presentation.
\begin{remark*}
There are $\DD$-dimensional constraints which are not axial products of $\DD$ one-dimensional constraints. For example for $\DD{=}2$, the $\NAK$ constraint defined in Section~\ref{sec:introduction} is not an axial product.
\end{remark*}
\end{proposition}
\begin{proof}
Let $\G_{S_1}=((V_1,E_1),\ELL_1),\G_{S_2}=((V_2,E_2),\ELL_2),{\ldots},\G_{S_\DD}=((V_\DD,E_\DD),\ELL_\DD)$ be  presentations of $S_1,{\ldots},S_\DD$, respectively. Define the $\DD$-tuple of labeled graphs $\barGG=(\G_1,{\ldots},\G_\DD)$, as follows. Let 
$$
E=\left\{\left(e_1,{\ldots},e_\DD\right){\in}\prod_{i=1}^{\DD}E_i:\ELL_1(e_1){=}\ELL_2(e_2){=}{\ldots}{=}\ELL_\DD(e_\DD)\right\},
$$
and let $\ELL:E\rightarrow\Sigma$ be given by
$$
\ELL(e_1,{\ldots},e_\DD)=\ELL_1(e_1)\;\;;\;\;(e_1,{\ldots},e_\DD){\in}E.
$$
For $i=1,2,{\ldots},\DD$, the graph $\G_i$ is defined by $\G_i=(G_i,\ELL)$ with $G_i=(V_i,E)$, where for $\blde=(e_1,{\ldots},e_\DD){\in}E$, $\sigma_{\G_i}(\blde)=\sigma_{\G_{S_i}}(e_i)$ and $\tau_{\G_i}(\blde)=\tau_{\G_{S_i}}(e_i)$. It's easy to verify that $\barGG$ is a presentation of $S_1{\otimes}S_2{\otimes}{\ldots}{\otimes}S_\DD$. 

Assume now that every $G_{S_i}$ is losseless. We show that in this case $\barGG$ is capacity preserving. Let $X{=}\opX(G_1){\otimes}{\ldots}{\otimes}\opX(G_\DD)$, $n$ be a positive integer,
and let $\bldn$ be the $\DD$-tuple with every entry equal to $n$.
We extend the mapping $\ELL$ to $\ELL:X_\bldn{\rightarrow}S_\bldn$ as described above. 
Now, fix an array $\Gamma{\in}S_\bldn$, and for $i=1,2,{\ldots},\DD$ let $\Gamma'^{(i)}{\in}(E_i)^\bldn$ be an array such that every row in direction $i$, $(\Gamma'^{(i)}_{\bldj_k})_{k=1}^{n}$ is a path in $\G_{S_i}$ generating the corresponding row $(\Gamma_{\bldj_k})_{k=1}^{n}$ in $\Gamma$. Let $\Gamma'{\in}E^\bldn$ be the array with entries given by $\Gamma'_\bldj=(\Gamma'^{(1)}_\bldj,{\ldots},\Gamma'^{(\DD)}_\bldj)$, $\bldj{\in}[n]^\DD$.
It follows from the construction of $\barGG$ that $\Gamma'{\in}X_\bldn$, and that $\ELL(\Gamma')=\Gamma$. Moreover, any array $\Delta{\in}X_\bldn$ such that $\ELL(\Delta)=\Gamma$ can be constructed in this manner.
Now, as each $\G_{S_i}$ is lossless, there are at most $|V_i|^2$ possibilities of choosing each row in direction $i$ of $\Gamma'^{(i)}$, and as there are $n^{\DD{-}1}$ such rows, there are at most $|V_i|^{2n^{\DD{-}1}}$ possibilities of choosing each $\Gamma'^{(i)}$. It follows that 
$$
|\ELL^{-1}(\{\Gamma\})|\leq\prod_{i=1}^{\DD}|V_i|^{2n^{\DD{-}1}},
$$
where $\ELL^{-1}(\{\Gamma\})=\{\Gamma'{\in}X_\bldn:\ELL(\Gamma')=\Gamma\}$. 
Summing the latter inequality over all $\Gamma{\in}\S_\bldn$, we obtain
$$
|X_\bldn|{\leq}|S_\bldn|\prod_{i=1}^{\DD}|V_i|^{2n^{\DD{-}1}}.
$$
Taking the $\log$, dividing through by $n^\DD$, and taking the limit as $n$ approaches infinity, we have $\capac(X){\leq}\capac(S)$. Clearly, $\capac(X){\geq}\capac(S)$, since $X$ is a presentation of $S$. The result follows.
\end{proof}

A graph $G=(V,E)$  is {\em irreducible} if
for any pair of vertices $u,v{\in}V$ there is a path in $G$ starting at
$u$ and terminating at $v$; it is {\em primitive} if it is irreducible and
the $\gcd$ of the lengths of all cycles of $G$ is $1$. We denote by $A(G)$
the adjacency matrix of $G$: namely  the $|V|\times|V|$ matrix where
$\left(A(G)\right)_{i,j}$ is the number of edges in $G$ from $i$ to $j$.
We use $\bldone$ to denote the real vector each of whose entries is $1$.
For a nonnegative matrix $A$ denote by $\lambda(A)$ its Perron eigenvalue,
that is, its largest real eigenvalue. It is well-known that for a
$1$-dimensional constraint $S$ presented by a lossless labeled graph $\G=(G,\ELL)$,
the capacity of $S$ is $\log\lambda\left(A(G)\right)$.
In particular, for a graph $G$, it holds that $\capac(\opX(G))=\log\lambda\left(A(G)\right)$.
We say that a graph $G$ is {\em symmetric} if $A(G)$ is symmetric. We say that a vertex of
a graph is {\em isolated} if it has no outgoing nor incoming edges.
We say that a vertex-constraint (resp. edge-constraint) is {\em symmetric} if
it is defined by a symmetric graph. For a vertex-constraint, this definition
is equivalent to requiring that the constraint is closed under reversal of the
order of symbols in words. Note, that in a symmetric edge-constraint, up to
removal of isolated vertices, the (unlabeled) graph defining the constraint is unique.

In this paper, we will mostly deal with $2$-dimensional constraints.
For a $2$-dimensional array $\Gamma{\in}\Sigma^{m_1{\times}m_2}$, for nonnegative integers $m_1$, $m_2$,
we denote by $\tr{\Gamma}$ its transpose, namely
$\left(\tr{\Gamma}\right)_{(i,j)}=\left(\Gamma\right)_{(j,i)}$ for all $(i,j){\in}[m_1]{\times}[m_2]$. For a $2$-dimensional constraint $S$ over $\Sigma$, we use $\tr{S}$ to denote the set
$$
\tr{S}=\left\{\Gamma{\in}\Sigma^{**}:\tr{\Gamma}{\in}S\right\},
$$
Clealry $\tr{S}$ is a $2$-dimensional constraint with $\capac(\tr{S})=\capac(S)$.

Let $S$ be a $2$-dimensional constraint over an alphabet $A$, and consider a horizontal (resp. vertical)
strip $\H_m(S)$ (resp. $\V_m(S)$) of $S$ for some positive integer $m$. We regard such a strip as a
set
of $1$-dimensional words over $\Sigma^m$ where each $m{\times}n$ (resp. $n{\times}m$) array in
the strip is considered a word of length $n$ over $A^m$.
Below we show that the horizontal and vertical strips of $S$ are $1$-dimensional constraints over $A^m$.
For this, we need the following definition. Let $G=(V,E)$ be a graph, and let $m$ be a positive integer.
Let $G^{{\times}m}$ be the graph given by $G^{{\times}m}=(V^m, E^m)$, where for
each $\blde=(e_1,{\ldots},e_m){\in}E^m$, $\sigma_{G^{{\times}m}}(\blde)=
(\sigma_G(e_1),{\ldots},\sigma_G(e_m))$ and
$\tau_{G^{{\times}m}}(\blde)=(\tau_G(e_1),{\ldots},\tau_G(e_m))$.
For a labeled graph
$\G=(G,\ELL)$ with $G=(V,E)$ and $\ELL:E\rightarrow\Sigma$, let $\G^{{\times}m}$ be
the labeled graph defined by $\G^{{\times}m}=(G^{{\times}m},\ELL^{{\times}m})$,
where $\ELL^{{\times}m}:E^m\rightarrow\Sigma^m$ is given by
$$
\begin{array}{ll}
\ELL^{{\times}m}(e_1,{\ldots},e_m)=(\ELL(e_1),{\ldots},\ELL(e_m))\;\;;\;\;(e_1,{\ldots},e_m){\in}E^m
\end{array}
$$
We call $G^{{\times}m}$ (resp., $\G^{{\times}m}$) the {\em $m$th tensor-power of $G$}.
(resp.,  $\G$). We can now state the following proposition.
\begin{proposition}
\label{prop:strips-are-1dim-constraints}
Let $S$ be a $2$-dimensional constraint over $\Sigma$ and let $m$ be a positive integer.
Then
\begin{enumerate}
\item \label{itm:strips-are-1dim-constraints:general}
$\H_m(S)$ (resp. $\V_m(S)$) is a $1$-dimensional constraint over $\Sigma^m$.
\item \label{itm:strips-are-1dim-constraints:axial}
Let $S=T^{(\V)}{\otimes}T^{(\H)}$ for $1$-dimensional constraints $T^{(\V)}$, $T^{(\H)}$
over $\Sigma$, presented by
labeled graphs $\G^{(\V)}$, $\G^{(\H)}$, respectively. Then the $1$-dimensional
constraint $\H_m(S)$ is presented by the labeled graph
$\G^{(\H)}_m$ defined as the sub-graph of the labeled
graph $(\G^{(\H)})^{{\times}m}$ consisting of
only those edges whose label (an $m$-letter word over $\Sigma$)
satisfies $T^{(\V)}$.  An analogous statement holds for $\V_m(S)$, with respect to
the graph $\G^{(\V)}_m$ formed in a similar way from $(\G^{(\V)})^{{\times}m}$.
\end{enumerate}
\end{proposition}
\begin{proof}
It suffices to prove this only for horizontal strips $\H_m(S)$.
We first prove part~{\ref{itm:strips-are-1dim-constraints:axial}}.
It's easy to verify that the labeled graph $(\G^{(\H)})^{{\times}m}$
presents the constraint over $\Sigma^m$, consisting of all $m{\times}n$ arrays
of $\Sigma^{**}$ such that every row satisfies $T^{(\H)}$. It follows that the
sub-graph $\G^{(H)}_m$, formed by removing all the edges of $(\G^{(\H)})^{{\times}m}$
that are labeled with a word that does not satisfy $T^{(\V)}$, presents the $1$-dimensional
constraint consisting of all $m{\times}n$ arrays with every row satisfying $T^{(H)}$ and
every column satisfying $T^{(\V)}$. This is precisely the constraint $\H_m(S)$.

We now prove part~\ref{itm:strips-are-1dim-constraints:general}. Let the pair of labeled
graphs $(\G^{(\V)},\G^{(\H)})$ be a presentation of $S$, where $\G^{(\V)}=((V^{(\V)},E),\ELL)$
and $\G^{(\H)}=((V^{(\H)},E),\ELL)$.
Define the edge-constraints $\E^{(\V)}=\opX(V^{(\V)},E)$ and
$\E^{(\H)}=\opX(V^{(\H)},E)$. Since $(\G^{(\V)},\G^{(\H)})$ is a presentation of $S$, we have
$S=\ELL(\E^{(\V)}{\otimes}\E^{(\H)})$, and therefore
$\H_m(S)=\ELL(\H_m(\E^{(\V)}{\otimes}\E^{(\H)}))$. By
part~\ref{itm:strips-are-1dim-constraints:axial}, $\H_m(\E^{(\V)}{\otimes}\E^{(\H)})$ is a
$1$-dimensional constraint, presented by a labeled graph $\G^{(\H)}_m$ with edges labeled
by words in $E^m$. Replacing each such label $\blde{\in}E^m$ in that graph with
$\ELL(\blde)$ we clearly obtain a presentation of $\ELL(\H_m(\E^{(\V)}{\otimes}\E^{(\H)})) = \H_m(S)$.
\end{proof}

We shall use the following notation for $2$-dimensional arrays. Let $A$ be a set and
$\Gamma{\in}A^{s{\times}t}$. For integers $0{\leq}s_1{\leq}s_2{<}s$ and $0{\leq}t_1{\leq}t_2{<}t$,
we denote by $\Gamma_{s_1:s_2,t_1:t_2}$ the
sub-array:
$$
\left(\Gamma_{s_1:s_2,t_1:t_2}\right)_{i,j}=\Gamma_{s_1{+}i,t_1{+}j}\;\;;\;\;
(i,j){\in}[s_2{-}s_1{+}1]{\times}[t_2{-}t_1{+}1],
$$
and by $\Gamma_{s_1:s_2,*}$ (resp. $\Gamma_{*,t_1:t_2}$) the sub-array
$\Gamma_{s_1:s_2,0:t-1}$ (resp. $\Gamma_{0:s-1,t_1:t_2}$).
We also abbreviate $x{:}x$ in the subscript by $x$.
We shall use the same notation for one-dimensional vectors:
for a vector $\bldv{\in}A^s$, $\bldv_{s:t}$ denotes the
subvector
$$
\left(\bldv_{s:t}\right)_i=\bldv_{s+i}\;\;;\;\;i{\in}[t{-}s{+}1].
$$

\section{Constraints with symmetric edge-constrained strips.}
\label{sec:edge-sym}
In this section we generalize the method presented in \cite{Calkin-Wilf} to provide improved lower bounds on
capacities of $2$-dimensional constraints whose horizontal strips are symmetric edge-constraints.

Fix an alphabet $\Sigma$, and let $S$ be a $2$-dimensional constraint over $\Sigma$. We say that $S$ has
{\em horizontal edge-constrained-strips}
if for every positive integer $m$,
the constraint $\H_m(S)$  is an edge-constraint.
If, in addition, every horizontal strip
is symmetric,
we say that $S$ has {\em symmetric horizontal
edge-constrained strips}. Analogously, using $\V_m(S)$,  we have the notions of a  $2$-dimensional constraint
with vertical edge-constrained-strips and  symmetric vertical edge-constrained-strips

Here, we consider constraints of the form $S=T\otimes{\E}$, where $\E{=}\opX(G_\E)$ is an
edge-constraint defined by the graph $G_\E=(V_\E,E_\E)$ and
$T$ is an arbitrary $1$-dimensional constraint over $\Sigma$. Then $\E$ is presented by
$\G_\E=(G_\E,I_E)$ where $I_E$ is the identity map on $E_\E$. Let $m$ be a positive integer.
By Proposition~\ref{prop:strips-are-1dim-constraints},
part~\ref{itm:strips-are-1dim-constraints:axial}, $\H_m(S)$ is a $1$-dimensional constraint
presented by a subgraph $\G^{(\H)}_m=(G^{(H)}_m,I_E^{{\times}m} )$ of $\G_\E^{{\times}m}$.
It follows that $\H_m(S)=\opX(G^{(\H)}_m)$, and so $S$
has horizontal edge-constrained strips. Henceforth, we further assume that
it has symmetric horizontal edge-constrained strips; note that symmetry of the
graph $G_\E$ is necessary but not sufficient for this assumption (see
Proposition~\ref{prop:edge-sym-cond-1} below).

For a positive integer $m$, let $F_m=|(V_\E)^m|$, and let $H_m$
denote the $F_m\times F_m$ adjacency matrix of $G^{(\H)}_m$.
Since the limit in~(\ref{eq:capacity}) exists independently of the choice of $(\bldm_i)_{i=1}^{\infty}$, and since
$\lim_{n\rightarrow\infty}\log(|S_{m\times n}|)/n=\capac(\H_m(S))$ for every positive integer $m$, we have
\begin{eqnarray}
\nonumber\capac(S)&=&\lim_{m,n\rightarrow\infty}\frac{\log|S_{m{\times}n}|}|{mn}\\
\nonumber         &=&\lim_{m\rightarrow\infty}\lim_{n\rightarrow\infty}\frac{\log|S_{m\times n}|}{mn} \\
\nonumber     &=&\lim_{m\rightarrow\infty}\frac{\capac(\H_m(S))}{m} \\
\label{eq:strip-cap-limit}   &=&\lim_{m\rightarrow\infty}\frac{\log\lambda(H_m)}{m}.
\end{eqnarray}
For a matrix $M$, let $\tr{M}$ denote its transpose. Fix a positive integer $m$.
Following \cite{Calkin-Wilf}, since $H_m$ is real and symmetric, we obtain by the
min-max principle~\cite{HJ}
$$
\lambda(H_m^p) \geq \frac{\tr{\bldy_m}H_m^p\bldy_m}{\tr{\bldy_m}\bldy_m},
$$
for any $F_m{\times}1$ real vector $\bldy_m\neq{\mathbf 0}$ and positive integer $p$.
Choosing $\bldy_m$ to be the vector $H_m^q\bldx_m$, for some positive integer $q$ and $F_m \times 1$ real vector $\bldx_m$ such that $\bldy_m\neq{\mathbf 0}$, we have
\begin{equation}
\label{eq:lambda_Hm_lower_bound}
\lambda(H_m^p) \geq \frac{\tr{\bldx_m}H_m^{2q+p}\bldx_m}{\tr{\bldx_m}H_m^{2q}\bldx_m}.
\end{equation}
Thus by~(\ref{eq:strip-cap-limit}), it follows that
\begin{equation}
\label{eq:capac-rayleigh-lb}
\capac(S)\geq \frac{1}{p}\limsup_{m\rightarrow\infty}\frac{1}{m}\log{\frac{\tr{\bldx_m}H_m^{2q+p}\bldx_m}{\tr{\bldx_m}H_m^{2q}\bldx_m}}.
\end{equation}

In \cite{Calkin-Wilf}, each $\bldx_m$ is chosen to be the $F_m\times 1$ vector with
every entry equal to $1$.
We obtain improved lower bounds in many cases by choosing other sequences of vectors,
$(\bldx_m)_{m=1}^{\infty}$, as follows.
We fix integers $\mu{\geq}0$ and $\alpha{\geq}1$, and
let ${\phi:(V_\E)^{\mu+\alpha}\rightarrow [0,\infty)}$ be a nonnegative function. Our method works
for sequences $\left(\bldx_{m_k}\right)_{k=1}^{\infty}$=$\left(\bldz^\phi_{m_k}\right)_{k=1}^{\infty}$
where $m_k=\mu+k\alpha$ for positive integers $k$, and
$\bldz^\phi_{m_k}$ is the $F_{m_k}{\times}1$ nonnegative vector indexed by $(V_\E)^{m_k}$:
\begin{equation}
\label{eq:bldzmk-def}
(\bldz_{m_k}^\phi)_{\bldv}{=}\prod_{i=0}^{k-1}\phi(\bldv_{i\alpha:i\alpha+{\mu+\alpha-1}})\;\;\; ; \;\;\; \bldv{\in}(V_\E)^{m_k}.
\end{equation}
For such sequences and a fixed positive integer $n$,
we will show that one can compute $L_n$, the growth rate of
$\tr{\bldx_{m_k}}H_{m_k}^n\bldx_{m_k}$:
$$
L_n=\lim_{k{\rightarrow}{\infty}}\frac{\log\tr{\bldx_{m_k}}H_{m_k}^n\bldx_{m_k}}{m_k},
$$
and from~(\ref{eq:capac-rayleigh-lb}) we obtain the lower bound $\capac(S){\geq}(L_{2q+p}-L_{2q})/p$.

Before doing this for general $\mu$ and $\alpha$, it is instructive to look at the special case:
$\mu=0$ and $\alpha=1$. In this case $m_k=k$, and
$$
\left(\bldz_{k}^\phi\right)_{\bldv}{=}\prod_{i=0}^{k-1}\phi((\bldv)_i)\;\;\; ; \;\;\;\bldv{\in}(V_\E)^{k}.
$$
Let $(\bldx_m)_{m=1}^{\infty}=(\bldz_m^\phi)_{m=1}^{\infty}$ and let $n$ be a positive integer. For a word $w=w_1{\ldots}w_n{\in}\E$ define its weight, $\W_\phi(w)$, by $\W_\phi(w){=}\phi(\sigma(w_1))\phi(\tau(w_n))$, where $w_1,w_n$ are regarded as edges in $G_\E$
and extend this to arrays $\Gamma{\in}S_{m{\times}n}$ by
$$
\W_\phi(\Gamma)=\prod_{i=0}^{m-1}\W_\phi(\Gamma_{i,:}).
$$
Observe that for an array $\Gamma{\in}S_{m{\times}n}$ that is a path in $G_m^{(H)}$ of length $n$
starting at $\bldv{\in}(V_\E)^{m}$ and ending at $\bldu{\in}(V_\E)^{m}$, it holds that
$\W_\phi(\Gamma){=}(\bldz_m^\phi)_\bldv(\bldz_m^\phi)_\bldu$. It follows that
\begin{equation}
\label{eq:sum-of-arrays-weights}
\tr{\bldx_m}H_m^n\bldx_m{=}\hspace{-8pt}\sum_{\Gamma{\in}S_{m{\times}n}}\hspace{-8pt}\W_\phi(\Gamma).
\end{equation}
Now, pick a deterministic presentation, $\G_n^{(\V)}{=}(V_n^{(\V)},E_n^{(\V)},\ELL_n^{(\V)})$ of $\V_n(S)$,
and let $\W_\phi:E{\rightarrow}[0,{\infty})$ be the edge weighting defined by $\W_\phi(e)=\W_\phi(\ELL_n^{(\V)}(e))$ for $e{\in}E_n^{(\V)}$.
Let $A(\G_n^{(\V)},\W_\phi)$ be the $|V_n^{(\V)}|{\times}|V_n^{(\V)}|$ weighted adjacency matrix of $\G_n^{(\V)}$ with entries indexed by $V_n^{(\V)}{\times}V_n^{(\V)}$ and given by
$$
\left(A(\G_n^{(\V)},\W_\phi)\right)_{i,j}=\hspace{-15pt}\mathop{\sum_{e\in E_n^{(\V)}:}}_{\sigma(e)=i,\tau(e)=j}\hspace{-15pt}\W_{\phi}(e)\;\;\;;\;\;\; i,j{\in}V_n^{(\V)}.
$$
Then
\begin{eqnarray*}
\tr{\bldone}A(\G_n^{(\V)},\W_\phi)^m\bldone
&=&\sum_{\gamma}\W_\phi(\ELL_n^{(\V)}(\gamma)),
\end{eqnarray*}
where the sum is taken over all paths $\gamma$ in $\G_n^{(\V)}$ of length $m$ and $\ELL_n^{(\V)}(\gamma)$ denotes
the array in $S_{m{\times}n}$ generated by $\gamma$. Since $\G_n^{(\V)}$ is deterministic it follows that
$$
\lim_{m{\rightarrow}\infty}\frac{\log\tr{\bldone}A(\G_n^{(\V)},\W_\phi)^m\bldone}{m}=
\lim_{m{\rightarrow}\infty}\frac{\log{\sum_{\Gamma{\in}S_{m{\times}n}}\W_\phi(\Gamma)}}{m}.
$$
By equation (\ref{eq:sum-of-arrays-weights}) the RHS is $L_n$ and by Perron-Frobenius theory the LHS is $\log\lambda(A(\G_n^{(\V)},\W_\phi))$, and thus $L_n=\log\lambda(A(\G_n^{(\V)},\W_\phi))$.

For general $\mu,\alpha$, we proceed similarly. We pick a deterministic presentation $\G_n^{(\V)}=(V_n^{(\V)},E_n^{(\V)},\ELL_n^{(\V)})$ of $\V_n(S)$ and construct a labeled directed graph $\I{=}\I(\mu,\alpha,n,\G_n^{(V)},G_\E){=}(V_\I,E_\I,\ELL_\I)$, with nonnegative real weights on its edges given by $\W_\phi:E_\I\rightarrow[0,\infty)$. The graph $\I$ and weight function $\W_\phi$ are defined as follows.
The set of vertices $V_\I$ is given by
$$
V_\I=\left\{(\bldf,v,\bldl)\::\:v\in V_n^{(\V)}, \bldf,\bldl\in (V_\E)^\mu \right\},
$$
and the function $\ELL_\I:E_\I\rightarrow\Sigma^{\alpha{\times}n}$ labels each edge with an $\alpha{\times}n$ array over $\Sigma$.
We specify the edges of $\I$ by describing the outgoing edges of each of its vertices along with
their weights.
Let $\bldv=(\bldf,v,\bldl)\in V_\I$ be a vertex of $\I$. The set of outgoing edges of $\bldv$ consists of exactly one edge for every path of length $\alpha$ in $\G_n^{(\V)}$ starting at $v$. Let $\gamma=\left(e_i\right)_{i=0}^{\alpha-1}\subseteq E_n^{(\V)}$ be such a path and let $u$ be its terminating vertex. We regard the word generated by $\gamma$ in $\G_n^{(\V)}$ as an array $\Gamma\in\Sigma^{\alpha \times n}$ with entries given by $(\Gamma)_{i,j}=\left(\ELL_n^{(\V)}(e_i)\right)_j$. Let $\bldf=(f_0,\ldots,f_{\mu-1})$ and $\bldl=(l_0,\ldots,l_{\mu-1})$ and for $i=\mu,\mu{+}1,{\ldots},\mu{+}\alpha{-}1$, define $f_i$ to be $\sigma(\Gamma_{i-\mu,0})$ and $l_i$ to be $\tau(\Gamma_{i-\mu,n-1})$, where $\Gamma_{i-\mu,0}$ and $\Gamma_{i-\mu,n-1}$ are regarded as edges in the graph $G_\E$. For such a path $\gamma$ the corresponding outgoing edge $e\in E_\I$ of $\bldv$ satisfies $\sigma(e)=\bldv$, $\ELL_\I(e)=\Gamma$, $\tau(e)=((f_{\alpha},f_{\alpha+1},\ldots,f_{\alpha+\mu-1}),u,(l_{\alpha},l_{\alpha+1},\ldots,l_{\alpha+\mu-1}))$.
The weight of $e$, $\W_{\phi}(e)$, is given by $\W_{\phi}(e)=\phi(f_0,\ldots,f_{\mu+\alpha-1})\phi(l_0,\ldots,l_{\mu+\alpha-1})$.
We shall regard the label of a path $(e_i)_{i=0}^{\ell{-}1}$ in $\I$ as the $\ell\alpha{\times}n$ array $\Gamma$ over $\Sigma$ resulting from concatenating the labels of the edges of $\gamma$ in order in the vertical direction, namely $\Gamma_{i\alpha+k,j}=(\ELL_\I(e_i))_{k,j}$, for all $i{\in}[\ell]$, $k{\in}[\alpha]$ and $j{\in}[n]$.
Finally, we define the weighted adjacency matrix of the labeled directed graph $\I$ with weights
given by $\W_\phi$ as the $|V_\I|\times|V_\I|$ nonnegative real matrix $A(\I,\W_\phi)$ with entries indexed by $(V_\I)^2$ and given by
$$
\left(A(\I,\W_\phi)\right)_{i,j}=\hspace{-15pt}\mathop{\sum_{e\in E_\I:}}_{\sigma(e)=i,\tau(e)=j}\hspace{-15pt}\W_{\phi}(e)\;\;\;;\;\;\; i,j\in V_\I.
$$

The following lemma generalizes ideas in~\cite{Calkin-Wilf} and uses the weighted labeled graph $\I$ to compute
$L_n$, when $(\bldx_{m_k})_{k=1}^{\infty}=(\bldz^\phi_{m_k})_{k=1}^{\infty}$.
\begin{lemma}
\label{lemma:xHx-growth-rate}
For $\left(\bldx_{m_k}\right)_{k=1}^{\infty}=\left(\bldz^\phi_{m_k}\right)_{k=1}^{\infty}$, and $\I=\I(\mu,\alpha,n,\G^{(\V)}_n,G_\E)$,
$$
\lim_{k\rightarrow\infty}\frac{\log\tr{\bldx_{m_k}}H_{m_k}^n\bldx_{m_k}}{m_k}=
\frac{\log\lambda(A(\I,\W_\phi))}{\alpha}.
$$
\end{lemma}
\begin{proof}
We shall show that there are positive real constants $c,d$ such that for all positive integers $k$,
\begin{equation}
\label{eq:xHmkx-bounds}
{c\cdot\tr{\bldone}\left(A(\I,\W_\phi)\right)^{k+\lceil\mu/\alpha\rceil}\bldone} \leq
{\tr{\bldx_{m_k}}H_{m_k}^n\bldx_{m_k}} \leq
{d\cdot\tr{\bldone}\left(A(\I,\W_\phi)\right)^k\bldone}.
\end{equation}

For a positive integer $s$ and vector $\blde=(e_0,\ldots,e_{s-1})$ in $\left(E_\E\right)^s$ denote by $\sigma(\blde),\tau(\blde)\in \left(V_\E\right)^s$ the vectors with entries given by
$$
\begin{array}{l}
\left(\sigma(\blde)\right)_i=\sigma(e_i)\\
\left(\tau(\blde)\right)_j=\tau(e_i)
\end{array}\;;\;
i\in\{0,1,\ldots,s{-}1\}.
$$

Now, fix a positive integer $k$. Let $\Gamma$ be an array in $S_{m_k{\times}n}$.
Recall that each entry of $\Gamma$ is an edge in $G_\E$ and define the weight of $\Gamma$, denoted $\W_{\phi}(\Gamma)$, by
$$
\W_{\phi}(\Gamma)=\prod_{i=0}^{k-1}{
\phi(\sigma(\Gamma_{i\alpha:i\alpha+\mu+\alpha-1,0}))
\phi(\tau(\Gamma_{i\alpha:i\alpha+\mu+\alpha-1,n-1}))
}.
$$
For a positive integer $\ell$, let $\P_\ell$ denote the set of paths of length $\ell$ in $\I$. We denote the label of a path $\gamma{\in}\P_\ell$ by $\ELL_\I(\gamma)$. It is easily verified that there exists a path in $\P_\ell$ with label $\Gamma\in\Sigma^{\ell\alpha\times n}$ if and only if there exists a path in $\G_n^{(\V)}$ of length $\ell\alpha$ that generates $\Gamma$. As $\G_n^{(\V)}$ is a presentation of $\V_n$, the set of labels of paths in $\P_\ell$ is $S_{\ell\alpha\times n}$.

For a finite path $\gamma$ in $\I$, define its weight, denoted $\W_{\phi}(\gamma)$, as the product of the weights of the edges in the path. Recalling that the entries of $\bldx_{m_k}=\bldz_{m_k}^\phi$ are indexed by
$(V_\E)^{m_k}$, we observe that
\begin{eqnarray*}
\tr{\bldx}_{m_k} H_{m_k}^n \bldx_{m_k}&=&\sum_{\Gamma\in S_{m_k{\times}n}}
\hspace{-8pt}{\left(\bldx_{m_k}\right)_{\sigma(\Gamma_{*,0})}\left(\bldx_{m_k}\right)_{\tau(\Gamma_{*,n{-}1})}}\\
&=&\sum_{\Gamma\in S_{m_k{\times}n}}\hspace{-8pt}{\W_{\phi}(\Gamma)}.
\end{eqnarray*}
For an array $\Gamma\in S_{m_k{\times}n}$, we say that a path $\gamma\in\P_k$ {\em matches} $\Gamma$ if it
is labeled by the sub-array $\Gamma_{\mu:m_k-1,*}$ and starts at a vertex $(\bldf,v,\bldl)\in V_\I$ with $\bldf=\sigma(\Gamma_{0:\mu-1,0})$ and $\bldl=\tau(\Gamma_{0:\mu-1,n-1})$. 
It can be verified from the construction of $\I$ that if $\gamma$ matches $\Gamma$ then $\W_{\phi}(\gamma)=\W_{\phi}(\Gamma)$.

Now, since $\G_n^{(\V)}$ is a presentation of $\V_n(S)$, it follows from the construction of $\I$ that every $\Gamma\in S_{m_k{\times}n}$ has a path in $\P_k$ matching it.
Conversely, since for a path $\gamma\in\P_k$ all arrays $\Gamma\in S_{m_k{\times}n}$ that it matches have the same
sub-array $\Gamma_{\mu:m_k-1,*}$, it follows that there are at most $|\Sigma|^{\mu n}$ arrays in $S_{m_k\times n}$ that $\gamma$ matches.
Therefore,
\begin{eqnarray*}
\tr{\bldx}_{m_k} H_{m_k}^n \bldx_{m_k}&=&\sum_{\Gamma\in S_{m_k{\times}n}}\hspace{-8pt}\W_{\phi}(\Gamma)\\
& \leq &|\Sigma|^{\mu n}\sum_{\gamma\in\P_k}\W_{\phi}(\gamma)\\
& = &|\Sigma|^{\mu n}\tr{\bldone}\left(A(\I,\W_\phi)\right)^k\bldone.
\end{eqnarray*}
This shows the right inequality of~(\ref{eq:xHmkx-bounds}), we now turn to the left. Set $k'=k+\lceil\mu/\alpha\rceil$, $s=\lceil\mu/\alpha\rceil\alpha-\mu$, and let $\psi:\P_{k'}\rightarrow S_{m_k\times n}$ be given by
$$
\psi(\gamma)=\left(\ELL_\I(\gamma)\right)_{s:k'\alpha-1,*}\;;\;\gamma\in\P_{k'}.
$$
For a path $\gamma\in\P_{k'}$, with $\gamma=\left(e_i\right)_{i=0}^{k'-1}\subseteq E_\I$, its weight satisfies
$$
\W_{\phi}(\gamma)=\prod_{i=0}^{k'-1}\W_{\phi}(e_i)=\left(\prod_{i=0}^{\lceil\mu/\alpha\rceil-1}\W_{\phi}(e_i)\right)\W_{\phi}(\psi(\gamma))
      \leq{\Phi}^{2\lceil\mu/\alpha\rceil}\W_{\phi}(\psi(\gamma)),
$$
where we take $\Phi$ to be a positive constant satisfying $\Phi{\geq}\max\left\{\phi(\bldv){:}\bldv{\in}(V_\E)^{\mu{+}\alpha}\right\}$.
Now let $\Gamma$ be an array in $S_{m_k\times n}$.
Since $\G_n^{(\V)}$ is deterministic, so is $\I$, and thus for every vertex $v\in V_\I$, the paths in $\P_k'$ starting at $v$ are labeled distinctly. As all paths $\gamma$ that map to $\Gamma$ under $\psi$ have the same sub-array $\left(\ELL_\I(\gamma)\right)_{s:k'\alpha-1,*}$, it follows that there are at most $|\Sigma|^{sn}$ paths $\gamma\in P_k'$ starting at $v$ such that $\psi(\gamma)=\Gamma$. Consequently, there are at most $|V_\I||\Sigma|^{sn}$ paths in $\P_{k'}$ that map to $\Gamma$ under $\psi$. Therefore,
\begin{eqnarray*}
\tr{\bldone}\left(A(\I,\W_\phi)\right)^{k+\lceil\mu/\alpha\rceil}\bldone
&= &\sum_{\gamma\in\P_{k'}}{\W_\phi(\gamma)}\\
&\leq& \Phi^{2\lceil\mu/\alpha\rceil}\sum_{\gamma\in\P_{k'}}{\W_\phi(\psi(\gamma))}\\
&\leq& \Phi^{2\lceil\mu/\alpha\rceil}|V_\I||\Sigma|^{sn}
\sum_{\Gamma\in S_{m_k{\times}n}}\hspace{-8pt}\W_{\phi}(\Gamma)\\
& = & \Phi^{2\lceil\mu/\alpha\rceil}|V_\I||\Sigma|^{sn}\,\tr{\bldx}_{m_k} H_{m_k}^n \bldx_{m_k}.
\end{eqnarray*}
Dividing both sides by $\Phi^{2\lceil\mu/\alpha\rceil}|V_\I||\Sigma|^{sn}$ we obtain the left inequality
of~(\ref{eq:xHmkx-bounds}).
The claim of the lemma now follows from Perron-Frobenius theory by taking the $\log$ of~(\ref{eq:xHmkx-bounds}), dividing it by $m_k$ and taking the limit as $k$ approaches infinity.
\end{proof}

We thus obtain the following lower bound on the capacity of a $2$-dimensional constraint.

\begin{theorem}
\label{thm:cap-lower-bound-edge}
Let $T,\E$ be $1$-dimensional constraints over an alphabet $\Sigma$, with $\E$ an edge constraint defined by a graph $G_\E=(V_\E,E_\E)$. Set $S=T{\otimes}\E$ and suppose that $S$ has symmetric horizontal edge-constrained strips. Let $\mu{\geq}0$ and $\alpha,p,q{>}0$ be integers and $\phi:\left(V_\E\right)^{\mu+\alpha}\rightarrow[0,\infty)$ be a nonnegative real function. For a positive integer $n$
let $\G_n$ be a deterministic presentation of $\V_n(S)$ and set $A_{n,\phi}=A(\I(\mu,\alpha,n,\G_n,G_\E),\W_\phi)$. Then
\begin{equation}
\label{eqn:cap-lower-bound}
\capac(S)\geq\frac{\log\lambda(A_{2q{+}p,\phi})-\log\lambda(A_{2q,\phi})}{p\alpha}.
\end{equation}
\begin{remark}
\label{rem:cw-upper-bounds}
In addition to computing lower-bounds, \cite{Calkin-Wilf} gives a method for computing upper bounds on the capacity of the hard-square constraint. It can be shown that this method can also be applied to all constraints of the form $T{\otimes}\E$, with $\E$ an edge constraint, having symmetric horizontal edge-constrained strips.
\end{remark}
\begin{remark}
\label{general}
Theorem~\ref{thm:cap-lower-bound-edge} can be generalized to apply to
$2$-dimensional constraints having symmetric horizontal edge-constrained strips,
which are not necessarily axial-products. Let $S$ be such a constraint, and for
every positive integer $m$, let
$G^{(\H)}_m=(V^{(\H)}_m,E^{(\H)}_m)$ be the symmetric graph, with no isolated vertices,
defining $\H_m(S)$. Set $G_\E=G^{(\H)}_1$, $V_\E=V^{(\H)}_1$ and $E_\E=E^{(\H)}_1$. We claim there
exists a mapping $f_m:V^{(\H)}_m{\rightarrow}(V_\E)^m$ such that for every edge
$\blde=e_0e_1{\ldots}e_{m{-}1}{\in}E^{(\H)}_m$, with each $e_i{\in}E_\E$,
\begin{equation}
\label{eq-class-map}
f_m(\sigma(\blde))=(\sigma(e_0),{\ldots},\sigma(e_{m{-}1}))
\mbox{ and }
f_m(\tau(\blde))=(\tau(e_0),{\ldots},\tau(e_{m{-}1})).
\end{equation}
This mapping is defined as follows. For a vertex $v{\in}V^{(\H)}_m$, pick an incoming edge $\blde=e_0e_1{\ldots}e_{m{-}1}{\in}E^{(\H)}_m$ and define $f_m(v)$ as $(\tau(e_0),{\ldots},\tau(e_{m{-}1}))$. This mapping is uniquely-defined: indeed if $\blde'=e'_0{\ldots}e'_{m{-}1}{\in}E^{(\H)}_m$ is another incoming edge of $v$, and $\bldg=g_0{\ldots}g_{m{-}1}{\in}E^{(\H)}_m$ is an outgoing edge of $v$, then clearly, for every $i{\in}[m]$, both $e_i g_i$ and $e'_i g_i$ are paths in $G_\E$; consequently $\tau(e_i)=\tau(e'_i)$.
It is easy to check that $f_m$ satisfies the conditions in (\ref{eq-class-map}).
Now, replace the definition of $\bldz_{m_k}^{\phi}$ in~(\ref{eq:bldzmk-def}) with
$$
(\bldz_{m_k}^\phi)_{\bldv}{=}\prod_{i=0}^{k-1}
\phi(\bldu_{i\alpha:i\alpha+{\mu+\alpha-1}})\;\;\; ; \;\;\;
\bldu=f_{m_k}(\bldv),\bldv{\in}V^{(\H)}_{m_k}.
$$
With this new definition and the aid of (\ref{eq-class-map}), it can be verified that
Lemma~\ref{lemma:xHx-growth-rate} and consequently Theorem~\ref{thm:cap-lower-bound-edge} still hold.
\end{remark}
\begin{remark}
Clearly, it is sufficient, for the theorem to hold, that $\H_m(S)$ is symmetric for large enough $m$.
\end{remark}
\end{theorem}

We now give a sufficient condition for the constraint $S=T{\otimes}\E$ to have symmetric horizontal edge-constrained strips.
For this to happen, we (generally) must have that $G_\E$ is symmetric. This means
that there exists a ``matching'' between edges, were each edge is matched with an edge
in the ``reverse'' direction. More precisely there is a bijection $
{R:E_\E\rightarrow E_\E}$ such that for all $e\in E_\E$, $(\sigma(e),\tau(e))=(\tau(R(e)),\sigma(R(e)))$
and $R(R(e))=e$. We call such a bijection an {\em edge-reversing matching}, and we denote
by $\R(G_\E)$ the set of all edge-reversing matchings of $G_\E$. Clearly a
graph $G$ is symmetric iff it has an edge-reversing matching. Thus
$T{\otimes}\E$ has symmetric horizontal edge-constrained strips iff
for every $m$, $\G_m^{(\H)}$ has an edge-reversing matching. We present a sufficient condition for this to hold.
\begin{proposition}
\label{prop:edge-sym-cond-1}
Let $T,\E$ be $1$-dimensional constraints over an alphabet $\Sigma$, with $\E$ an edge constraint defined by a graph $G_\E=(V_\E,E_\E)$ with $R\in\R(G_\E)$ an edge-reversing matching. If for all positive integers $m$ and words $e_1{\ldots}e_m{\in}E_\E^m$, if $e_1{\ldots}e_m$ satisfies $T$ then $R(e_1){\ldots}R(e_m)$ satisfies $T$ as well, then
$T{\otimes}\E$ has symmetric horizontal edge-constrained strips.
\end{proposition}
\begin{proof}
Let $G^{(\H)}_m=(V^m_\E,E^{(\H)}_m)$ be the sub-graph of
$G_\E^{{\times}m}$ that defines $\H_m(S)$.
We show that $G^{(\H)}_m$ is symmetric. Let ${R^{{\times}m}:E_\E^m\rightarrow E_\E^m}$ be defined
by
$$
R^{{\times}m}(e_1,{\ldots},e_m)=(R(e_1),{\ldots},R(e_m)).
$$
Clearly, $R^{{\times}m}$ is an edge-reversing matching of $G_\E^{{\times}m}$.
Recall that $E^{(H)}_m$ consists of all the edges in $E_\E^m$ that,
when regarded as $m$-letter words over $\Sigma$, satisfy $T$.
Therefore, by the assumption, it follows that for all $\blde{\in}E^{(H)}_m$, $R^{{\times}m}(\blde){\in}E^{(H)}_m$ as well. Consequently, $R^{{\times}m}$ restricted to $E^{(H)}_m$, is an edge-reversing
matching of $G^{(\H)}_m$ and hence it is symmetric.
\end{proof}
If $\G_T$ is a presentation of $T$ and $R{\in}\R(G_\E)$, a sufficient condition for the hypothesis of 
Proposition~\ref{prop:edge-sym-cond-1} to hold, which may be easier to check, is the existance of a function $f:E_T\rightarrow E_T$ satisfying: 1) $\ELL_T(f(e))=R(\ELL_T(e))$ for all $e\in E_T$ and 2) for any path $e_1 e_2$ of length $2$ in $\G_T$, the sequence $f(e_1) f(e_2)$ is also a path in $\G_T$. Indeed, if such a function exists then any path $\epsilon_1\epsilon_2{\ldots}\epsilon_m$ in $G_T$ generating a word $e_1e_2{\ldots}e_m$ has a corresponding path $f(\epsilon_1)f(\epsilon_2){\ldots}f(\epsilon_m)$ generating the word $R(e_1)R(e_2){\ldots}R(e_m)$, and thus the hypothesis of Proposition~\ref{prop:edge-sym-cond-1} is fullfiled. In fact, it can be shown that when $\G_T$ is irreducible, deterministic and has the minimum number of states  among all deterministic presentations of $T$, this condition is also necessary for the hypothesis of Propositon~\ref{prop:edge-sym-cond-1} to hold (see~\cite[Section 3.3]{Lind-Marcus}).

In Section~\ref{sec:sofic-prod} we use Proposition~\ref{prop:edge-sym-cond-1} to show that the method described in this
section can be used to compute lower bounds on $\CHARGE(b_1){\otimes}\CHARGE(b_2)$ for any
positive integers $b_1$ and $b_2$.

\section{Constraints with symmetric vertex-constra\-ined strips.}
In this section we present an analog to Theorem~\ref{thm:cap-lower-bound-edge} that gives
lower bounds on the capacities of constraints for which every horizontal or every vertical
strip is a symmetric vertex-constraint. We do this by transforming a $2$-dimensional
constraint with symmetric vertex-constrained strips to a $2$-dimensional constraint with
symmetric edge-constrained strips, having the same capacity.

Fix an alphabet $\Sigma$, and let $S$ be a $2$-dimensional constraint over $\Sigma$. We say
that $S$ has {\em horizontal vertex-constrained strips} if for every positive
integer $m$, the constraint $\H_m(S)$ is a vertex-constraint. If, in addition, 
every horizontal strip is symmetric, we say that $S$ has {\em symmetric horizontal vertex-constrained strips}.
The notions of a $2$-dimensional constraint with vertical vertex-constrained strips and symmetric vertical vertex-constrained strips are defined analogously. 

It turns out that $\RWIM$ and $\NAK$ do not have horizontal or
vertical edge-constrained strips, and so the method
in section~\ref{sec:edge-sym} does not apply directly.  We illustrate this only for horizontal
strips for $S=\RWIM$.  Recall from section~\ref{sec:framework} that an edge
constraint is a constraint of memory 1 such that
any two follower sets are either disjoint or identical.
We claim that this condition does not hold for $\H_m(S)$.  To see this,
given any $m$, let $w$ be the all-zeros word of length
$m$ and $u \ne w$ be any other word of length $m$.
Now, the $m \times 2$ arrays
$$
\begin{array}{cc}
w_0 & w_0\\
w_1 & w_1 \\
\ldots & \ldots \\
w_{m-1} & w_{m-1}
\end{array}, ~~~
\begin{array}{cc}
w_0 & u_0\\
w_1 & u_1 \\
\ldots & \ldots \\
w_{m-1} & u_{m-1}
\end{array}, ~~~
\begin{array}{cc}
u_0 & w_0\\
u_1 & w_1 \\
\ldots & \ldots \\
u_{m-1} & w_{m-1}
\end{array},
$$
belong to $S$, yet the $m \times 2$ array
$$
\begin{array}{cc}
u_0 & u_0\\
u_1 & u_1 \\
\ldots & \ldots \\
u_{m-1} & u_{m-1}
\end{array},
$$
does not.
Thus, $w$ and $u$, viewed as $m \times 1$ columns have different but non-disjoint
follower sets.  Thus, $\RWIM$ does not have horizontal edge-constrained strips.

However, it is not hard to show that $\RWIM$ and $\NAK$ have both
symmetric horizontal vertex-constrained strips and symmetric vertical vertex-constrained
strips. For instance, for $S = \RWIM$, $\H_m(S)$ is the vertex constraint defined by
the graph $G = (V,E)$, where $V$ consists of all binary vectors $u_0 \ldots u_{m-1}$ of length
$m$ and $E$ consists of a single edge from
$u \in V$ to $v \in V$
iff for all $i$, whenever $u_i =1$, then $v_{i+1} = v_i = v_{i-1} = 0$ (with the obvious
modification when $i=0$ or $m-1$).  And
$\V_m(S)$ is the vertex constraint defined by the graph of $G' = (V',E')$,
where $V'$ consists of all binary vectors $u_0 \ldots u_{m-1}$ of length
$m$ which do not contain two adjacent `$1$'s and $E'$ consists of a single edge from
$u \in V$ to $v \in V$
iff for all $i$, whenever $u_i =1$, then $v_{i+1}= v_{i-1} = 0$ (again, with the obvious
modification when $i=0$ or $m-1$).  Clearly, both $G$ and $G'$ are symmetric.

Now, let $S$ be a $2$-dimensional constraint over $\Sigma$. For a finite $m\times n$ array $\Gamma$ with $m\geq 1$ and $n\geq 2$ over $\Sigma$ its {\em $[1\times 2]$-higher block recoding} or {\em $[1\times 2]$-recoding} is an $m\times(n-1)$ array $\hat{\Gamma}$ over $\Sigma^{1{\times}2}$ with entries given by
$$
\hat{\Gamma}_{i,j}=\left(\Gamma_{i,j}\;\Gamma_{i,j+1}\right)\;;\; i=0,\ldots,m-1,\,j=0,\ldots,n-2.
$$
We denote by $S^{[1\times 2]}$ the set of all $[1\times 2]$-recodings of arrays in $S$
and refer to it as
the $[1\times 2]$-higher block recoding of $S$.
The $[1\times 2]$-higher block recoding of a constraint is a constraint. This is stated in the following proposition.
\begin{proposition}
Let $S$ be a $2$-dimensional constraint over $\Sigma$. Then $S^{[1\times2]}$ is a $2$-dimensional constraint over $\Sigma^{1{\times}2}$.
\begin{remark*} We may, of course, define, in a similar manner, the $[s\times t]$-higher block recoding of $S$, for any positive integers $s$ and $t$, and the $[s\times t]$-higher block recoding of a $2$-dimensional constraint $S$ is a $2$-dimensional constraint.
\end{remark*}
\end{proposition}
\begin{proof}
Set $S'=S^{[1{\times}2]}$ and let $(\G_{\V}, \G_{\H})$ be a presentation of $S$ with $\G_\V=(V_{\V},E,\ELL)$ and
$\G_\H=(V_\H,E,\ELL)$.
We construct labeled graphs $\G'_\V=(V_\V{\times}V_\V, E',\ELL')$ and $\G'_\H=(E,E',\ELL')$ as follows.
The set of edges $E'$ is defined as
$$
E'=\left\{(e_0\;e_1){\in}E^{1{\times}2}:\mbox{$e_0,e_1$ is a path in $\G^{(H)}$}\right\},
$$
and the labeling function $\ELL':E'\rightarrow\Sigma^{1{\times}2}$ is given by
$$
\ELL'(e_0\;e_1)=(\ELL(e_0)\;\ELL(e_1))\;\;;\;\;(e_0\;e_1){\in}E'.
$$
For every edge $(e_0\;e_1){\in}E'$ we define
$$
\begin{array}{lcl}
\sigma_{\G'_\V}(e_0\;e_1)&=&(\sigma_{\G_\V}(e_0),\sigma_{\G_\V}(e_1))\\
\tau_{\G'_\V}(e_0\;e_1)&=&(\tau_{\G_\V}(e_0),\tau_{\G_\V}(e_1))\\
\sigma_{\G'_\H}(e_0\;e_1)&=&e_0\\
\tau_{\G'_\H}(e_0\;e_1)&=&e_1
\end{array}.
$$
It's easy to verify that $(\G'_\V,\G'_\H)$ is a presentation of $S'$.
\end{proof}

Clearly, recoding is an injective mapping, thus $|S_{m\times n}|=|S^{[1\times 2]}_{m\times(n-1)}|$ for all positive integers $m\geq 1,n\geq 2$. It follows that $\capac(S)=\capac(S^{[1\times 2]})$.
The next proposition shows that the $[1\times2]$-higher block recoding of a constraint with symmetric horizontal vertex-constrained strips has symmetric horizontal edge-constrained strips.
\begin{proposition}
\label{prop:vert-to-edge}
Let $S$ be a $2$-dimensional constraint with horizontal vertex-constrained strips.
\begin{enumerate}
 \item \label{prop:vert-to-edge:itm:1} $S^{[1{\times}2]}$ has horizontal edge-constrained strips. Moreover,
$S^{[1{\times}2]}$  has symmetric horizontal edge-constrained strips iff $S$ has symmetric horizontal
vertex-constrained strips.
 \item \label{prop:vert-to-edge:itm:2} $S^{[1\times2]}=\V_1(S^{[1\times2]}){\otimes}\H_1(S^{[1\times2]})$.
\end{enumerate}
\end{proposition}
\begin{proof}
(\ref{prop:vert-to-edge:itm:1}). Let $m$ be a positive integer, $G^{(\H)}_m=(V^{(\H)}_m,E_m)$
be the graph defining the vertex-constraint $\H_m(S)$ and set $\Delta=\Sigma^{1{\times}2}$,
where $\Sigma$ is the alphabet of $S$. We define a labeling $\ELL:E_m\rightarrow\Delta^{m{\times}1}$
of the edges of $G^{(\H)}_m$. For $\blde{\in}E_m$, we
regard $\sigma(\blde)$ and $\tau(\blde)$ as $m{\times}1$
arrays over $\Sigma$, and define $\ELL(\blde)$ to be the array in $\Delta^{m{\times}1}$
with entries given by
$$
\ELL(\blde)_{(i,0)}=(\sigma(\blde)_{i,0}\;\tau(\blde)_{i,0})\;\;;\;\;i{\in}[m].
$$
It's easily verified that the word generated by every path in the labeled graph
$(G^{(\H)}_m,\ELL)$ is the $[1{\times}2]$-higher block recoding of the array formed
by concatenating the vertices along the path horizontally in sequence. It follows
that $(G^{(\H)}_m,\ELL)$ is a presentation of $\H_m(S^{[1{\times}2]})$. Since the labels of the
edges in $(G^{(\H})_m,\ELL)$ are distinct, we may identify each edge with its label, and it follows
that $\H_m(S^{[1{\times}2]})$ is an edge-constraint. Since the same graph defines both
$\H_m(S)$ and $\H_m(S^{[1{\times}2]})$ (the former as a vertex-constraint and the
latter as an edge-constraint), it follows that $\H_m(S)$ is symmetric iff $\H_m(S^{[1{\times}2]})$ is.
This completes the proof.

(\ref{prop:vert-to-edge:itm:2}). Clearly, $S^{[1\times2]}\subseteq \V_1(S^{[1\times2]}){\otimes}\H_1(S^{[1\times2]})$.
As for the reverse inclusion, let $\hat{\Gamma}\in\V_1(S^{[1\times2]}){\otimes}\H_1(S^{[1\times2]})$ be an $m\times n$ array
over $\Sigma^{1\times2}$. Since every row of $\hat{\Gamma}$ is in $\H_1(S^{[1\times2]})$, every row has a unique $1\times(n+1)$ pre-image under the recoding map. Let $\Gamma$ be the $m\times(n+1)$ array over $\Sigma$, whose $i$th row is the pre-image under the recoding map of the $i$th row of $\hat{\Gamma}$, for $i=0,\ldots,m-1$. Clearly, $\hat{\Gamma}$ is the $[1\times2]$-higher block recoding of $\Gamma$. Thus, it suffices to show that $\Gamma\in S$. For $i=0,1,\ldots,n-1$ clearly, the $m\times 2$ array $\Gamma_{*,i:i+1}$ over $\Sigma$ recodes to the column $\hat{\Gamma}_{*,i}$. By our assumption this column is in $\V_1(S^{[1\times2]})$. Since recoding is injective, $\Gamma_{*,i:i+1}$ must be in $S$. Since this holds for all $i=0,1,\ldots,n-1$ and since $\H_m(S)$ has memory $1$, it follows that $\Gamma\in S$ and therefore $\hat{\Gamma}\in S^{[1\times2]}$.
\end{proof}

We can now use the method described in Section~\ref{sec:edge-sym} to get
lower bounds on $2$-dimensional constraints with symmetric horizontal
vertex-constrained strips. This is stated in the following theorem.
\begin{theorem}
Let $S$ be a $2$-dimensional constraint over an alphabet $\Sigma$ with symmetric horizontal vertex-constrained
strips. Let $\mu{\geq}0$, and $\alpha,p,q{>}0$ be integers, $G_\E=(V_\E,E_\E)$ be the graph defining the vertex-constraint $\H_1(S)$ (hence $V_\E{\subseteq}\Sigma$), and $\phi:\left(V_\E\right)^{\mu{+}\alpha}\rightarrow[0,\infty)$ be a nonnegative function.
For an integer $n{\geq}2$, let $\G_n$ be a labeled graph obtained from a deterministic presentation for $\V_n(S)$ by replacing each edge-label with its $[1{\times}2]$-higher block recoding. Set $\hat{A}_{n,\phi}=A(\I(\mu,\alpha,n{-}1,\G_n,G_\E),\W_\phi)$, where $\I$, $\W_\phi$, and $A(\I, \W_\phi)$ are as defined in Section~\ref{sec:edge-sym}. Then
$$
\capac(S)\geq\frac{\log\lambda(\hat{A}_{p{+}2q{+}1,\phi})-\log\lambda(\hat{A}_{2q{+}1,\phi})}{p\alpha}.
$$
\end{theorem}
\begin{proof}
Let $S'=S^{[1{\times}2]}$. By Proposition~\ref{prop:vert-to-edge}, $S'=\V_1(S'){\otimes}\H_1(S')$, and $S'$ has horizontal
symmetric edge-constrained strips. Since $G_\E$ has no parallel edges, we may identifiy each edge $e{\in}E_\E$ with the pair $(\sigma(e),\tau(e))$; then, with this identification, $\H_1(S')=\opX(G_\E)$. Also, note that $\G_{2q+p+1}$ and $\G_{2q+1}$ are deterministic presentations for $\V_{2q+p}(S')$ and $\V_{2q}(S')$, respectively. The result follows from Theorem~\ref{thm:cap-lower-bound-edge} applied to $S'$.
\end{proof}

\section{Capacity bounds for axial products of constraints.}
\label{sec:sofic-prod}
In this section we show how the method described in Section~\ref{sec:edge-sym}
can be applied to axial products of certain
$1$-dimensional constraints. Let $S$ and $T$ be two $1$-dimensional constraints over an alphabet $\Sigma$. We wish to lower bound the capacity of the $2$-dimensional constraint $T{\otimes}S$. To this end, we pick a lossless presentation $\G_S=(G_S,\ELL_S)$, with $G_S=(V_S,E_S)$, for $S$. We extend the function $\ELL_S$ to multidimensional arrays over $E_S$ in the manner described in Section~\ref{sec:framework}, and for a set $A\subseteq\Sigma^*$, we denote by $\ELL_S^{-1}(A){\subseteq}E_S^*$ the inverse image of $A$ under this map, namely
$$
\ELL_S^{-1}(A)=\left\{w\in E_S^* : \ELL_S(w){\in}A\right\}.
$$
The following proposition shows that we can reduce the problem of calculating the capacity of $T{\otimes}S$ to that of calculating the capacity of $\ELL_S^{-1}(T){\otimes}\opX(G_S)$.
\begin{proposition}
\label{prop:lift}
Let $S, T$ be two $1$-dimensional constraints and let $\opX(G_S)$ and $\ELL_S^{-1}(T)$ be
 as defined above. Then
\begin{enumerate}
 \item \label{itm:lift-sofic} $\ELL_S^{-1}(T)$ is a $1$-dimensional constraint.
 \item \label{itm:lift-same-cap} $\capac(T{\otimes}S)=\capac(\ELL_S^{-1}(T){\otimes}\opX(G_S))$.
\end{enumerate}
\end{proposition}
\begin{proof}
\ref{itm:lift-sofic}. Let $\G_T=(V_T,E_T,\ELL_T)$ be a presentation of $T$. We shall construct a presentation $\F=(V_T,E_\F,\ELL_\F)$ of $\ELL_S^{-1}(T)$. The set of edges is given by
$E_\F=\{(e_T, e_S){\in}E_T{\times}E_S:\ELL_T(e_T){=}\ELL_S(e_S)\}$, and for an edge $(e_T,e_S){\in}E_\F$, $\sigma_{\F}(e_T,e_S)=\sigma_{\G_T}(e_T)$, $\tau_{\F}(e_T, e_S)=\tau_{\G_T}(e_T)$ and $\ELL_\F(e_T,e_S)=e_S$.
It is easily verified that $\ELL_S^{-1}(T)$ is presented by $\F$, and therefore it is a $1$-dimensional constraint.

\ref{itm:lift-same-cap}. We set $R=T{\otimes}S$, $U=\ELL_S^{-1}(T){\otimes}\opX(G_S)$. For an array
$\Delta{\in}R_{m{\times}n}$, define $P_\Delta=\{\Gamma{\in}U_{m{\times}n}:\ELL_S(\Gamma)=\Delta\}$, we claim that
\begin{equation}
\label{eq:Pdelta-bounds}
1{\leq}|P_\Delta|\leq|V_S|^{2m}.
\end{equation}
Indeed, it's easily verified that an array $\Gamma{\in}E_S^{m{\times}n}$ is in $P_\Delta$ iff for all $i{\in}[m]$ the row $(\Gamma_{i,j})_{j=0}^{n{-}1}$ is a path in $\G_S$ that generates $(\Delta_{i,j})_{j=0}^{n{-}1}$. Since $\G_S$ is a lossless presentation of $S$, for every $i{\in}[m]$, there is at least one path in $\G_S$ generating $(\Delta_{i,j})_{j=0}^{n{-}1}$ and at most $|V_S|^2$ such paths; the claim follows. Now, clearly for any $\Gamma{\in}U_{m{\times}n}$ the array $\ELL_S(\Gamma)$ is in $R_{m{\times}n}$. It follows that the sets $P_\Delta$, for $\Delta{\in}R_{m{\times}n}$ form a partition of $U_{m{\times}n}$, and we have
$$
\left|U_{m{\times}n}\right|=
\sum_{\Delta{\in}R_{m{\times}n}}|P_\Delta|.
$$
Therefore, by~(\ref{eq:Pdelta-bounds}), we get
$$
\left|R_{m{\times}n}\right|\leq
\left|U_{m{\times}n}\right|\leq
\left|R_{m{\times}n}\right||V_S|^{2m},
$$
and it follows from~(\ref{eq:capacity}) that $\capac(R)=\capac(U)$.
\end{proof}

Therefore if $\ELL_S^{-1}(T){\otimes}\opX(G_S)$ has symmetric horizontal edge-constrained strips, we can apply the method of Section~\ref{sec:edge-sym} to obtain lower bounds on $\capac(T{\otimes}S)$. In this case, it also follows from Remark~\ref{rem:cw-upper-bounds} of Theorem~\ref{thm:cap-lower-bound-edge}, that the method of \cite{Calkin-Wilf} for obtaining upper bounds on the capacity of the hard-square constraint, can be used to obtain upper bounds on $\capac(T{\otimes}S)$. Proposition~\ref{prop:edge-sym-cond-1} present a sufficient condition for $\ELL_S^{-1}(T){\otimes}\opX(G_S)$ to have symmetric horizontal edge-constrained strips.
Here we give another stronger sufficient condition involving only the presentation $\G_S$. We say that a labeled graph $(G,\ELL)$, with $G=(V,E)$, is {\em symmetric as a labeled graph}, if there exists an edge-reversing matching $R\in\R(G)$ which preserves $\ELL$, that is $\ELL(R(e))=\ELL(e)$ for all $e\in E$.
We assume now that $\G_S$ is symmetric as a labeled graph, and that $R\in\R(V_S,E_S)$ is an
edge-reversing matching which preserves $\ELL_S$. Since for any positive integer $m$ and $e_1{\ldots}e_m{\in}E_S^m$, the label $\ELL(e_1){\ldots}\ELL(e_m)=\ELL(R(e_1))\ldots\ELL(R(e_m))$, it follows that $e_1{\ldots}e_m{\in}\ELL_S^{-1}(T)$ iff $R(e_1){\ldots}R(e_m){\in}\ELL_S^{-1}(T)$. Consequently, the hypothesis of Proposition~\ref{prop:edge-sym-cond-1} holds and we have the following corollary.
\begin{corollary}
If $\G_S$ is symmetric as a labeled graph then $\ELL_S^{-1}(T){\otimes}\opX(G_S)$ has
symmetric horizontal edge-constrained strips.
\end{corollary}
Since the presentation in Figure~\ref{fig:even-presentation} is symmetric as a labeled graph,
we can apply the method of Section~\ref{sec:edge-sym} to get lower bounds on the capacity of all
constraints $T{\otimes}\EVEN$ for any $1$-dimensional
constraint $T$.

Let $S=\CHARGE(b_1)$ and let $T=\CHARGE(b_2)$ for integers $b_1,b_2\geq 2$.
Let $\G_S=(G_S,\ELL_S)$, with $G_S=(V_S,E_S)$ be the presentation given in
Figure~\ref{fig:chg-presentation} for $b=b_1$. Evidently, $\G_S$ is symmetric
with exactly one edge-reversing matching, $R:E_S{\rightarrow}E_S$. Fix a positive
integer $m$ and let $\blde=e_1e_2{\ldots}e_m{\in}E_S^m$.
Obviously,
$T$ is closed under negation of words (i.e., negating each symbol), and we have
$$
\begin{array}{cl}
               &e_1e_2{\ldots}e_m{\in}\ELL_S^{-1}(T)\\
\Longleftrightarrow&\ELL_S(e_1)\ELL_S(e_2){\ldots}\ELL_S(e_m){\in}T\\
\Longleftrightarrow&(-\ELL_S(e_1))(-\ELL_S(e_2)){\ldots}(-\ELL_S(e_m)){\in}T\\
\Longleftrightarrow&\ELL_S(R(e_1))\ELL_S(R(e_2)){\ldots}\ELL_S(R(e_m)){\in}T\\
\Longleftrightarrow&R(e_1)R(e_2){\ldots}R(e_m){\in}\ELL_S^{-1}(T).
\end{array}
$$
Consequently, it follows by Proposition~\ref{prop:edge-sym-cond-1}
that $\ELL_S^{-1}(T){\otimes}\opX(G_S)$ has symmetric horizontal
edge-constrained strips and we can apply the method of Section~\ref{sec:edge-sym}
to obtain lower bounds on the capacity of $\CHARGE(b_2)\otimes\CHARGE(b_1)$.

The reader will note a similarity in the constructions in proofs of
Propositions~\ref{axial-is-a-constraint}
and~\ref{prop:lift}.  Indeed, as an alternative approach, one may be able to use the
construction in
Proposition~\ref{axial-is-a-constraint} to obtain bounds on $\capac(S\otimes T)$:
namely, if $G_1$ and $G_2$ are the
underlying graphs of a capacity-preserving presentation $(\G_1,\G_2)$ of $S\otimes T$
and
$\opX(G_1) \otimes \opX(G_2)$ has  symmetric horizontal edge-constrained strips.
However, the approach given by Proposition~\ref{prop:lift} seems to be more direct and simpler than
the alternative approach.

\section{Heuristics for choosing $\phi$.}
\label{sec:heuristics}
In this section, we use the notation defined in Section~\ref{sec:edge-sym}, and assume that $S=T{\otimes}\E$ is a $2$-dimensional constraint with symmetric horizontal edge-constrained strips, where $\E$ is an edge constraint.
We describe heuristics for choosing the function $\phi$ to obtain ``good'' lower bounds on the capacity of $S$.

\subsection{Using max-entropic probabilites.}
\label{sec:heuristics-max-ent}
Recall that a vertex of a directed graph is isolated if no edges in the graph are connected to it. Note, that since $G^{(\H)}_m$ is symmetric, every vertex is either isolated or has both incoming and outgoing edges. We assume here that for every positive integer $m$, ignoring isolated vertices, $G^{(\H)}_m$ is a primitive graph.
In this case, the Perron eigenvector of $H_m$ is unique up to multiplication by a scalar. Let $\bldr_m$ be the right Perron eigenvector of $H_m$ normalized to be a unit vector in the $L_2$-norm. Observe, that substituting $\bldr_m$ for $\bldx_m$ satisfies~(\ref{eq:lambda_Hm_lower_bound}) with an equality. This motivates us to choose $\phi$ so that the resulting vector $\bldz_m^\phi$ approximates $\bldr_m$. Since $G^{(\H)}_m$ (without its isolated vertices) is irreducible, there is a unique stationary probability measure having maximum entropy on arrays of $\H_m$, namely the max-entropic probability measure on $\H_m$. We denote it here by ${\Pr}^{*,m}$. It is given by
$$
{\Pr}^{*,m}(\Gamma)=\frac{\left(\bldr_m\right)_{\sigma(\Gamma)}\left(\bldr_m\right)_{\tau(\Gamma)}}
{\lambda(H_m)^\ell}.
$$
for $\Gamma\in S_{m\times\ell}$, for some positive integer $\ell$, and where $\sigma(\Gamma),\tau(\Gamma)\in V_\E^m$
are given by
$$
\begin{array}{lcl}
\left(\sigma(\Gamma)\right)_i&=&\sigma(\Gamma_{i,0})\\
\left(\tau(\Gamma)\right)_i&=&\tau(\Gamma_{i,\ell-1})
\end{array}
\;\;;\;\;\;i=0,1,2,\ldots,m{-}1.
$$
Let $\bldV^{(m)}$ be a random variable taking values in $V_\E^m$ with distribution given by
$$
\Pr(\bldV^{(m)}{=}\bldv)={\Pr}^{*,m}\left\{\Gamma\in S_{m\times 1}:\sigma(\Gamma)=\bldv\right\} \;\;;\;\;\bldv\in V_\E^m.
$$
It's easily verified that
\begin{equation}
\label{eq:prob-pvec-rel}
\Pr(\bldV^{(m)}{=}\bldv)=\left((\bldr_m)_\bldv\right)^2.
\end{equation}

Thus approximating $\Pr(\bldV^{(m)}=\bldv)$ and taking a square root will give us an approximation for $(\bldr_m)_\bldv$.
Roughly speaking, $\Pr(\bldV^{(m)}=\bldv)$ is the probability of seeing the column of vertices $\bldv$
in the ``middle'' of an $m\times\ell$ array chosen uniformaly at random from $S_{m\times\ell}$, for large $\ell$.
Fix integers $\mu\geq0, \alpha\geq1$ as in Section~\ref{sec:edge-sym}, and assume now that $m=m_k=\mu+k\alpha$, for a positive integer $k$. For an integer $0{\leq}s{<}m$ and vectors $\bldu{\in}V_\E^\ell,\bldw{\in}V_\E^r$, with lengths satisfying $\ell{\leq}m{-}s, r{\leq}s$, denote by $p_s^{(m)}(\bldu)$ and $p_s^{(m)}(\bldu|\bldw)$ the probabilities given by
\begin{eqnarray*}
p_s^{(m)}(\bldu)&=&\Pr(\bldV^{(m)}_{s:s+\ell-1}=\bldu) \\
p_s^{(m)}(\bldu|\bldw)&=&\Pr(\bldV^{(m)}_{s:s+\ell-1}=\bldu|\bldV^{(m)}_{s-r:s-1}=\bldw).
\end{eqnarray*}
Then by the chain rule for conditional probability we have, for any vector $\bldv{\in}V_\E^m$,
$$
\Pr(\bldV^{(m)}{=}\bldv)=p_0^{(m)}(\bldv_{0:\mu{-}1})\prod_{i=0}^{k-1}p_{\mu{+}i\alpha}^{(m)}(\bldv_{\mu{+}i\alpha:\mu{+}(i{+}1)\alpha{-}1}|\bldv_{0:\mu{+}i\alpha{-}1}).
$$
A plausible way to approximate $\Pr(\bldV=\bldv)$, is by treating $\bldV$ as the outcome of a Markov process. Here we use a Markov process with memory $\mu$, and assume that $p_s(\bldu|\bldw)$ can be ``well'' approximated by $p_s(\bldu|\bldw_{r{-}\mu:r{-}1})$, for vectors $\bldu{\in}V_\E^\ell,\bldw{\in}V_\E^r$, with $r,\ell$ as above, and $r{\geq}\mu$. Using this approximation we get
$$
\Pr(\bldV^{(m)}{=}\bldv)\approx p_0^{(m)}(\bldv_{0:\mu{-}1})
\prod_{i=0}^{k-1}p_{\mu{+}i\alpha}^{(m)}(\bldv_{\mu{+}i\alpha:\mu{+}(i{+}1)\alpha{-}1}|\bldv_{i\alpha:\mu{+}i\alpha{-}1}).
$$

We hypothesize that for fixed vectors $\bldu{\in}V_\E^\alpha,\bldw{\in}V_\E^\mu$, as $m$ gets large, the conditional probabilities $p^{(m)}_{s}(\bldu|\bldw)$, for $0{\ll}s{\ll}m{-}1$, are ``approximately equal'' to the value when $s$ is in the ``middle'' of the interval $[0,m-1]$. We hypothesize that this holds for ``most'' of the integers $s$ in that inteval and moreover that this middle value converges as $m$ gets large.
Accordingly, we try to approximate the conditional probability $p^{(m)}_s(\bldu|\bldw)$ by the conditional probability found in the ``middle'' of a ``tall'' horizontal strip. More precisely, we fix an integer $\delta\geq 0$, set $\omega{=}2\delta{+}\mu{+}\alpha$, and approximate $p^{(m)}_s(\bldu|\bldw)$ by $p^{(\omega)}_{\delta+\mu}(\bldu|\bldw)$. We also approximate $p_0^{(m)}(\bldw)$ by $p_0^{(\omega)}(\bldw)$. This gives us
$$
\Pr(\bldV^{(m)}{=}\bldv)\approx p_0^{(\omega)}(\bldv_{0:\mu{-}1})
\prod_{i=0}^{k-1}p_{\delta+\mu}^{(\omega)}(\bldv_{\mu{+}i\alpha:\mu{+}(i{+}1)\alpha{-}1}|
\bldv_{i\alpha:\mu{+}i\alpha{-}1}),
$$
which, by~(\ref{eq:prob-pvec-rel}), implies that
\begin{equation}
\label{eq:approx-rm}
\left(\bldr_m\right)_\bldv\approx\sqrt{p_0^{(\omega)}(\bldv_{0:\mu{-}1})}
\prod_{i=0}^{k-1}\sqrt{p_{\delta+\mu}^{(\omega)}(\bldv_{\mu{+}i\alpha:\mu{+}(i{+}1)\alpha{-}1}|
\bldv_{i\alpha:\mu{+}i\alpha{-}1})}\;\;;\;\;\bldv{\in}V_\E^{m_k}.
\end{equation}
Set $F_m{=}|V_\E|^m$, and denote by $\widetilde{\bldr_{m_k}}{\in}\RR^{F_{m_k}}$ the nonnegative real vector with entries indexed by $V_\E^{m_k}$ and given by the RHS of equation~(\ref{eq:approx-rm}). Let $\phi:(V_\E)^{\mu+\alpha}\rightarrow[0,\infty)$ be given by
\begin{equation}
\label{eq:markov-prob-phi}
\phi(\bldu)=\sqrt{p_{\delta+\mu}^{(\omega)}(\bldu_{\mu:\mu+\alpha-1}|\bldu_{0:\mu-1})}\;\;;\;\;
\bldu{\in}(V_\E)^{\mu+\alpha},
\end{equation}
and let $\bldz_{m_k}^\phi\in\RR^{F_{m_k}}$ be the vector with entries indexed by $V_\E^{m_k}$ and given by~(\ref{eq:bldzmk-def}). Setting $\bldx_{m_k}=\bldz^\phi_{m_k}$, we obtain
$$
\left(\bldr_{m_k}\right)_\bldv\approx\left(\widetilde{\bldr_{m_k}}\right)_\bldv=
\left(\bldx_{m_k}\right)_\bldv\sqrt{p_0^{(\omega)}(\bldv_{0:\mu{-}1})}\;\;;\;\;\;\bldv{\in}(V_\E)^{m_k}.
$$
Now for $m_k{\geq}\omega$, if $\bldv{\in}V_\E^{m_k}$ is not an isolated vertex in $G_{m_k}$, then clearly, $\bldv_{0:\omega-1}$ is not an isolated vertex in $G_\omega$ as well. Therefore $(\bldr_\omega)_{\bldv_{0:\omega-1}}{>}0$, which implies that $p_0^{(w)}(\bldv_{0:\omega{-}1}){>}0$ and thus $p_0^{(\omega)}(\bldv_{0:\mu{-}1}){>}0$.
Let $p_{\min}=p_{\min}^{(w)}=\min\{p_0^{(\omega)}(\bldw):\mbox{$\bldw{\in}V_\E^\mu$ and $p_0^{(\omega)}(\bldw){>}0$}\}$. It follows that for all vertices $\bldv{\in}\left(V_\E\right)^{m_k}$ of $G_{m_k}$ that are not isolated, we have
$$
\left(\widetilde{\bldr_{m_k}}\right)_\bldv\geq \sqrt{p_{\min}} \left(\bldx_{m_k}\right)_\bldv.
$$
Now, for any positive integer $\ell$ and $F_{m_k}{\times}1$-real vector,
the product $\tr{\bldy}H_m^\ell\bldy$ depends only on the values of the entries of $\bldy$ indexed by non-isolated vertices of $G_{m_k}$. Consequently, we may write
$$
{p_{\min}}\tr{\bldx_{m_k}}H_{m_k}^\ell\bldx_{m_k}\leq
\tr{\widetilde{\bldr_{m_k}}}H_{m_k}^\ell\widetilde{\bldr_{m_k}}\leq
\tr{\bldx_{m_k}}H_{m_k}^\ell\bldx_{m_k},
$$
for all positive integers $\ell$. Taking the $\log$, dividing by $m_k$, and taking the limit as $k$ approaches infinity, we obtain
$$
\lim_{k\rightarrow\infty}{\frac{\log\tr{\widetilde{\bldr_{m_k}}}H_{m_k}^\ell\widetilde{\bldr_{m_k}}}{m_k}}=
\lim_{k\rightarrow\infty}{\frac{\log\tr{\bldx_{m_k}}H_{m_k}^\ell\bldx_{m_k}}{m_k}},
$$
where by Lemma~\ref{lemma:xHx-growth-rate}, the limit in the RHS exists.
Thus, choosing $\phi$ as given by~(\ref{eq:markov-prob-phi}) and computing the lower bound by the method
described in Section~\ref{sec:edge-sym} is equivalent to computing the limit of the lower bound in~(\ref{eq:lambda_Hm_lower_bound}), with $\widetilde{\bldr_m}$ substituted for $\bldx_m$, as $m$ approaches infinity. If $\widetilde{\bldr_m}$ approximates $\bldr_m$ well enough, we expect to get good bounds. Note that we may use the heuristic described here even for constraints for which the graphs $G^{(\H)}_m$ are not always irreducible. In this case, the geometric multiplicity of the Perron eigenvalue may be larger than $1$, and there may be more than one choice of the vector $\bldr_\omega$ in the computation of $p_{\delta{+}\mu}^{(\omega)}(\cdot|\cdot)$. Regardless of our choice, we will get a nonnegative function $\phi$ and a lower bound on the capacity. In Section~\ref{sec:numerical-results} we show numerical results obtained using the heuristic described here for several constraints.

\subsection{General optimization.}
\label{sec:heuristics-gen-opt}
We may also use general optimization techniques to find functions $\phi$ which maximize the lower bound on the capacity.
Fix integers $\mu{\geq}0$ and $p,q,\alpha{>}0$, and for a positive integer $\ell$, set $\D_{\ell}=(V_\E)^\ell$. In this subsection, we identify a function $\phi{:}\D_{\mu{+}\alpha}{\rightarrow}\RR$ with a real vector $\phi{\in}\RR^{|\D_{\mu+\alpha}|}$ indexed by $\D_{\mu{+}\alpha}$; for each $\bldj{\in}\D_{\mu{+}\alpha}$ we identify $\phi(\bldj)$ with the entry $\phi_\bldj$. 
For a positive integer $n$, let $\G_n$ be a deterministic presentation for $\V_n(S)$, $\I_n=\I(\mu,\alpha,n,\G_n,G_\E)$, and for a function $\phi:\D_{\mu+\alpha}{\rightarrow}[0,\infty)$, set $A_{n,\phi}=A(\I_n,\W_\phi)$.
Observe that for a scalar $c\in[0,\infty)$, $A_{n,c\phi}=c^2A_{n,\phi}$. It follows that using $c\phi$ in place of $\phi$ in equation~(\ref{eqn:cap-lower-bound}) of Thereom~\ref{thm:cap-lower-bound-edge} does not change the lower bound. Consequently, (as $\phi$ cannot be the constant $0$ function), it's enough to consider functions $\phi$ whose images (of all vectors in $(V_\E)^{\mu+\alpha}$) sum to $1$. We thus have the following optimization problem.
\begin{equation}
\label{eq:phi-optimization}
\begin{array}{ll}
\mbox{maximize}&
\left(\log\lambda(A_{2q{+}p,\phi})-\log\lambda(A_{2q,\phi})\right)/
\left(p\alpha\right)\\
\mbox{subject to}&
\phi\geq\bldzero,\\
&\phi\cdot\bldone=1,
\end{array}
\end{equation}
where $\bldzero$ and $\bldone$ denote the real vectors of size $|\D_{\mu+\alpha}|$ with every entry equal to
$0$ and $1$ respectively, and for two real vectors of the same size, $\bldt,\bldr$ we write $\bldt\geq\bldr$ or $\bldt{>}\bldr$ if the corresponding inequality holds entrywise.

Finding a global solution for a general optimization problem can be hard. We proceed to show that if we replace the constraint $\phi{\geq}\bldzero$ with $\phi{>}\bldzero$ in~(\ref{eq:phi-optimization}), thereby changing the feasable set and possibly decreasing the optimal solution, it can be formulated as an instance of a particular class of optimization problems known as ``DC optimization'' which may be easier to solve. Let $d$ be a positive integer.
A real-valued function $f:\RR^d\rightarrow\RR$ is called a {\em DC} (difference of convex) function, if it can be written
as the difference of two real-valued convex functions on $\RR^d$. An optimization problem of the form
$$
\label{eq:DC-optimization}
\begin{array}{ll}
\mbox{maximize}& f(x)\\
\mbox{subject to}& x{\in}X,\\
& h_i(x)\leq 0\;;\;\;i=0,1,\ldots,\ell,
\end{array}
$$
where $X\subseteq\RR^d$ is a convex closed subset of $\RR^d$ and the functions $f,h_0,\ldots,h_\ell$ are DC functions, is called a {\em DC optimization} or {\em DC programming} problem. See \cite{Horst-Thoai} and the references within for an overview of the theory of DC optimization.

A nonnegative function $f:\RR^d\rightarrow[0,\infty)$ is called {\em log-convex} or {\em superconvex}, if either $f(\bldt){>}0$ for all $\bldt{\in}\RR^d$ and $\log{f}$ is convex in $\RR^d$, or $f{\equiv}0$. A log-convex function is convex, and in~\cite{Kingman}, it is shown that the class of log-convex functions is closed under addition, multiplication, raising to positive real powers taking limits, and additionally that for a square matrix $A(\bldt)=(a_{i,j}(\bldt))$ whose entries are log-convex functions $a_{i,j}:\RR^d\rightarrow[0,\infty)$, the function $\bldt\rightarrow\lambda(A(\bldt))$ is log-convex as well.

Now, observe, that for a positive integer $n$, every entery of $A_{n,\phi}$ is a quadratic form in the entries $\phi(\bldj)$, $\bldj{\in}\D_{\mu+\alpha}$, with nonnegative integer coefficients. Such a function is generally not log-convex. To fix this, we perform the change of variables $\phi=e^\psi$, where $\psi$ is a real-valued function $\psi:\D_{\mu{+}\alpha}\rightarrow\RR$. Note that by doing so, we added the constraint $\phi{>}\bldzero$. Since every entry of $\phi$ is now positive, we may replace the constraint $\phi\cdot\bldone=1$ by the constraint $\phi(\bldv_0)=1$ or equivalently $\psi(\bldv_0)=0$, for some fixed $v_0{\in}\D_{\mu{+}\alpha}$. Problem~(\ref{eq:phi-optimization}) with the additional constraint $\phi{>}\bldzero$, can now be rewritten as
\begin{equation}
\label{eq:phi-optimization-DC}
\begin{array}{ll}
\mbox{maximize}&
\left(\log\lambda(A_{2q{+}p,{e^\psi}})-\log\lambda(A_{2q,{e^\psi}})\right)/
\left(p\alpha\right)\\
\mbox{subject to}&\psi(\bldv_0)=0.
\end{array}
\end{equation}
Obviously, we may substitue the maximization problem constraint, $\psi(\bldv_0)=0$, above into the objective function, thereby reducing the number of variables by $1$; however, this is not relevant for the
discussion, so, for simplicity, we do not do so here. Now, for a positive integer $n$, the entries of the matrix $A_{n,e^{\psi}}$ are of the form
$$
\sum_{k=1}^{q_{i,j}}{e^{\psi(\bldw_{k,i,j})+\psi(\bldu_{k,i,j})}},
$$
where $q_{i,j}$ are nonnegative integers, and $\bldw_{k,i,j}$ and $\bldu_{k,i,j}$ are vectors in $\D_{\mu{+}\alpha}$, for all $i,j{\in}V_{\I_n}$ and integers $1{\leq}k{\leq}q_{i,j}$. It can be verified that a function of such a form is log-convex in $\psi$. It follows that the function  $\psi{\rightarrow}\lambda(A_{n,e^{\psi}})$ for $\psi\in\RR^{|\D_{\mu{+}\alpha}|}$ is log-convex as well. Therefore either $\lambda(A_{n,e^{\psi}}){\equiv}0$, or $\lambda(A_{n,e^{\psi}}){>}0$ for all $\psi\in\RR^{|\D_{\mu{+}\alpha}|}$. In particular, for $\psi{\equiv}0$ the matrix $A_{n,\bldone}$ is the adjacency matrix of the graph $\I_n$. Since $\I_n$ is deterministic, and for every nonnegative integer $\ell$, the set of labels of its paths of length $\ell$ is $S_{\ell\alpha{\times}n}$, it follows that $(1{/}\alpha)\log\lambda(A_{n,\bldone}){=}\capac(\V_n(S))\geq\capac(S)$ (the latter inequality follows from~(\ref{eq:capac-inf})). Hence, if $\capac(S){>}{-}\infty$ (or equivalently $\capac(S){\geq}0$), then $\lambda(A_{n,e^\psi}){>}0$ for all $\psi\in\RR^{|\D_{\mu+\alpha}|}$ and $\log\lambda(A_{n,e^\psi})$ is a convex function of $\psi$. 
Clearly $\capac(S){>}{-}\infty$ iff. $S_{m{\times}n}{\neq}\emptyset$, for all positive integers $m,n$. 
We thus obtain the following theorem.
\begin{theorem}
Let $S$ be a constraint such that for all positive integers $m,n$, $S_{m{\times}n}{\neq}\emptyset$ then
Problem~\textup{(\ref{eq:phi-optimization-DC})} is a DC optimization problem.
\end{theorem}

\section{Numerical results for selected constraints.}
\label{sec:numerical-results}
In this section we give numerical lower bounds on the capacity of some
$2$-dimensional constraints obtained using
the method presented in the sections above. The constraints considered
are $\NAK$, $\RWIM$, $\EVEN\osqr$, and $\CHARGE(3)\osqr$. Table~\ref{tbl:best-bounds}
summarizes the best lower bounds obtained using
 our method.
For comparison, we provide the best lower bounds that we could obtain using other methods.
We also give upper bounds on the capacity of these constraints obtained using the
method of~\cite{Calkin-Wilf}. Table \ref{tbl:results-max-ent} shows the lower bounds obtained
using our max-entropic probability heuristic for choosing $\phi$,
described in Section~\ref{sec:heuristics-max-ent}. Table \ref{tbl:results-gen-opt}
shows the lower bounds obtained with our method by trying to solve the optimization problem
described in Section~\ref{sec:heuristics-gen-opt}. In this, we did not make use of the DC
property of the optimization problem; instead, we used a generic sub-optimal optimization algorithm
whose results are not guarenteed to be global solutions. Utilizing special algorithms for solving
DC optimization problems may give better lower bounds. The rightmost column of each of these tables
shows the lower bound calculated for the same values of $p$ and $q$ using the method of
\cite{Calkin-Wilf}. The largest lower-bound obtained for each constraint is marked with a `$^\star$'. 
In the next subsections we give remarks specific to some of these constraints.

The numerical results were computed using the eigenvalue routines in Matlab and rounded (down for lower bounds and up for upper-bounds) to $10$ decimal places. Given accuracy problems with possibly defective matrices, we verified the results using the technique described in~\cite[Section IV]{Nagy-Zeger}. 

\subsection{The constraint $\RWIM$}
Observe that this constraint has both symmetric horizontal and symmetric vertical vertex-constrained strips. Thus, we can apply our method in the vertical as well as the horizontal direction to get lower bounds. Clearly, $\capac(\tr{\RWIM})=\capac(\RWIM)$, so we can obtain additional lower bounds on $\capac(\RWIM)$ by using our method to get lower bounds on $\capac(\tr{\RWIM})$. Some of these bounds are given in Tables~\ref{tbl:results-max-ent} and \ref{tbl:results-gen-opt}.

\subsection{The constraint $\EVEN\osqr$}
We used the reduction described in Section~\ref{sec:sofic-prod} with $\G_{\EVEN}$ being the presentation of $\EVEN$ given in Figure~\ref{fig:even-presentation}, to get lower bounds on the capacity of $\EVEN\osqr$. Table~\ref{tbl:results-gen-opt} gives the results obtained with our method using the optimization described in Section~\ref{sec:heuristics-gen-opt}. We also used the method with the max-entropic probability heuristic of Section~\ref{sec:heuristics-max-ent} and the results are given in Table~\ref{tbl:results-max-ent}.

\subsection{The constraint $\CHARGE(b)\osqr$}
For this constraint, the case $b{=}1$ is degenerate. Indeed, there are exactly two $m{\times}n$ arrays in $\CHARGE(1)\osqr$ for all positive integers $m$ and $n$, and consequently, $\capac(\CHARGE(1)\osqr){=}0$. For $b{=}2$, we show in Theorem~\ref{thm:capac-chg2} in Section~\ref{sec:exact-comp} that the capacity is exactly $1/4$.
For $b{=}3$, we used the reduction of Section~\ref{sec:sofic-prod} with $\G_{\CHARGE(3)}$ being the presentation of $\CHARGE(3)$ given in Figure~\ref{fig:chg-presentation}, to get lower bounds on the capacity of $\CHARGE(3)\osqr$. Table~\ref{tbl:results-gen-opt} gives the results obtained with our method using the optimization described in Section~\ref{sec:heuristics-gen-opt}. We also used the method with the max-entropic probability heuristic of Section~\ref{sec:heuristics-max-ent} and the results are given in Table~\ref{tbl:results-max-ent}.

\begin{table}[ht]
\caption{Best known lower bounds on capacities of certain constraints.}
\vspace{-10pt}
\begin{center}
\begin{tabular}{|l|l|c|c|}
\hline
Constraint & Prev. best lower bound & New lower bound & Upper bound \\
\hline
$\NAK$            & $0.4250636891^\star$  & $0.4250767745$   & $0.4250767997^\star$\\
$\RWIM$           & $0.5350150^\dagger$   & $0.5350151497$   & $0.5350428519^\star$\\
$\EVEN\osqr$      & $0.4385027973^\star$  & $0.4402086447$   & $0.4452873312^\star$\\
$\CHARGE(3)\osqr$ & $0.4210209862^\star$  & $0.4222689819$   & $0.5328488954^\star$\\
\hline
\end{tabular}
\end{center}
\label{tbl:best-bounds}
$^\star$Calculated using the method of~\cite{Calkin-Wilf}.\\
$^\dagger$Appears in~\cite{Yong-Golin}.
\end{table}

\begin{table}[ht]
\caption{Lower bounds using max-entropic probabilty heuristic (Section~\ref{sec:heuristics-max-ent}).}
\vspace{-10pt}
\begin{center}
$$
\begin{array}{|c|l|l|l|l|l|l|l|}
\hline
\mbox{Constraint} & \delta& \mu & \alpha & p & q & \mbox{Lower bound} & \mbox{Using \cite{Calkin-Wilf}}\\
\hline
\NAK    & 3     & 1     & 1      & 1 & 5 & 0.4250766244   & 0.4248771038 \\
        & 3     &   1   & 1      & 2 & 4 & 0.4250766446   & 0.4249055702 \\
        & 6     &   1   & 1      & 1 & 5 & 0.4250767227   & 0.4248771038 \\
        & 3     &   3   & 4      & 1 & 5 & 0.4250767590   & 0.4248771038 \\
        & 7     &   1   & 1      & 1 & 5 & 0.4250767617   & 0.4248771038 \\
        & 3     &   1   & 1      & 2 & 6 & 0.4250767647   & 0.4250636891 \\
        & 3     &   1   & 4      & 1 & 5 & 0.4250767733   & 0.4248771038 \\
        & 5     &   1   & 1      & 1 & 5 & 0.4250767744   & 0.4248771038 \\
        & 3     &   1   & 4      & 2 & 6 & 0.4250767745   & 0.4250636891 \\
        & 3     &   1   & 2      & 2 & 6 & 0.4250767745   & 0.4250636891 \\
\hline
\RWIM & 3 & 1 & 3 & 1 & 6 & 0.5350147968 & 0.5235145644\\
     & 1 & 1 & 1 & 3 & 6 & 0.5350148753 & 0.5318753627\\
     & 3 & 2 & 2 & 1 & 5 & 0.5350148814 & 0.5160533001\\
     & 3 & 1 & 2 & 2 & 6 & 0.5350149069 & 0.5337927416\\
     & 2 & 1 & 2 & 2 & 6 & 0.5350149071 & 0.5337927416\\
     & 0 & 1 & 2 & 2 & 6 & 0.5350149136 & 0.5337927416\\
     & 2 & 2 & 2 & 1 & 5 & 0.5350149271 & 0.5160533001\\
     & 1 & 1 & 2 & 2 & 6 & 0.5350149462 & 0.5337927416\\
     & 1 & 1 & 3 & 1 & 6 & 0.5350149525 & 0.5235145644\\
     & 1 & 1 & 1 & 1 & 7 & 0.5350149707 & 0.5280406048\\
\hline
\tr{\RWIM} & 4 & 1 & 3 & 2 & 4 & 0.5350145937 & 0.5350144722\\
       & 1 & 1 & 1 & 1 & 5 & 0.5350146612 & 0.5350149478\\
       & 4 & 2 & 1 & 1 & 4 & 0.5350147212 & 0.5350142142\\
       & 3 & 1 & 1 & 1 & 5 & 0.5350147328 & 0.5350149478\\
       & 5 & 1 & 1 & 1 & 5 & 0.5350147619 & 0.5350149478\\
       & 2 & 2 & 1 & 1 & 4 & 0.5350147969 & 0.5350142142\\
       & 4 & 1 & 1 & 1 & 5 & 0.5350148255 & 0.5350149478\\
       & 2 & 1 & 1 & 1 & 5 & 0.5350148449 & 0.5350149478\\
       & 0 & 1 & 1 & 1 & 5 & 0.5350148814 & 0.5350149478\\
       & 0 & 2 & 1 & 1 & 4 & 0.5350148980 & 0.5350142142\\
\hline
\EVEN\osqr & 3 & 2 & 1 & 1 & 3 & 0.4383238232 & 0.4347423815\\
       & 3 & 1 & 1 & 1 & 4 & 0.4383243738 & 0.4367818624\\
       & 3 & 1 & 3 & 2 & 3 & 0.4383632350 & 0.4356897662\\
       & 3 & 1 & 2 & 4 & 3 & 0.4383838005 & 0.4364303826\\
       & 3 & 1 & 1 & 2 & 4 & 0.4384647082 & 0.4371709990\\
       & 3 & 1 & 3 & 3 & 3 & 0.4384906740 & 0.4360537982\\
       & 3 & 1 & 2 & 1 & 4 & 0.4385448358 & 0.4367818624\\
       & 3 & 1 & 2 & 2 & 4 & 0.4386655840 & 0.4371709990\\
       & 3 & 1 & 3 & 1 & 4 & 0.4387455520 & 0.4367818624\\
\hline
\CHARGE(3)\osqr & 0 & 0 & 1 & 1 & 2 & 0.4188210386 & 0.4101473707\\
                & 0 & 0 & 1 & 1 & 4 & 0.4222689819^\star & 0.4197053158\\
\hline
\end{array}
$$
\end{center}
$^\star$Best lower bound.
\label{tbl:results-max-ent}
\end{table}

\begin{table}[ht]
\caption{Lower bounds using optimization (Section~\ref{sec:heuristics-gen-opt}).}
\begin{center}
$$
\begin{array}{|l|l|l|l|l|l|l|}
\hline
\mbox{Constraint}& \mu & \alpha & p & q & \mbox{Lower bound} & \mbox{Using \cite{Calkin-Wilf}}\\
\hline

\NAK  & 2 &	1 &	2 &	4 &	0.4250767692 &	0.4249055702 \\
      & 1 & 	2 & 	1 & 	5 & 	0.4250767736 &	0.4248771038 \\
      & 1 & 	1 &	3 &	4 &	0.4250767737 &	0.4248960814 \\
      & 1 &	2 &	1 &	3 &	0.4250767739 &	0.4224650194 \\
      & 1 &	1 &	4 &	4 &	0.4250767739 &	0.4249674993 \\
      & 1 &	1 &	5 &	4 &	0.4250767740 &	0.4249783192 \\
      & 1 &	1 &	6 &	4 &	0.4250767741 &	0.4249995626 \\
      & 1 &	2 &	3 &	3 &	0.4250767743 &	0.4244240822 \\
      & 1 &	2 &	6 &	3 &	0.4250767744 &	0.4247979797 \\
      & 1 &     2 &     2 &     5 &     0.4250767745^{\star} & 0.4250294285\\
\hline
\RWIM  & 1   & 1         & 1 & 3 & 0.5350147515 & 0.4832292495\\
       & 1   & 1         & 2 & 3 & 0.5350148675 & 0.5300373650\\
       & 1   & 1         & 3 & 3 & 0.5350149371 & 0.5212673183\\
       & 1   & 1         & 1 & 4 & 0.5350150805 & 0.5037272248\\
       & 1   & 1         & 2 & 4 & 0.5350151001 & 0.5318663054\\
       & 1   & 1         & 3 & 4 & 0.5350151123 & 0.5265953036\\
       & 1   & 1         & 1 & 5 & 0.5350151372 & 0.5160533001\\
       & 1   & 1         & 2 & 5 & 0.5350151410 & 0.5330440001\\
       & 1   & 1         & 2 & 6 & 0.5350151491 & 0.5337927416\\
       & 1   & 2         & 1 & 5 & 0.5350151497^{\star} & 0.5160533001\\
\hline
\tr{\RWIM} 
       & 1   &  2    & 4 & 3 &  0.5350151364 & 0.5350130576 \\
       & 1   &  2    & 3 & 4 &  0.5350151377 & 0.5350146307 \\
       & 1   &  2    & 5 & 3 &  0.5350151392 & 0.5350134356 \\
       & 2   &  1    & 1 & 4 &  0.5350151405 & 0.5350142142 \\
       & 1   &  1    & 1 & 5 &  0.5350151442 & 0.5350149478 \\
       & 1   &  2    & 1 & 4 &  0.5350151465 & 0.5350142142 \\
       & 1   &  2    & 1 & 5 &  0.5350151481 & 0.5350149478 \\
       & 1   &  2    & 2 & 4 &  0.5350151482 & 0.5350144722 \\
       & 1   &  3    & 1 & 4 &  0.5350151483 & 0.5350142142 \\
\hline
\EVEN\osqr & 1   &  1    &  1 & 3 & 0.4395381520 & 0.4347423815 \\
       & 1   &  1    &  2 & 3 & 0.4397347451 & 0.4356897662 \\
       & 1   &  1    &  1 & 4 & 0.4402086447^{\star} & 0.4367818624 \\
\hline
\CHARGE(3)\osqr & 0   & 1   &  1 & 2 &  0.4189237100  & 0.4101473707 \\
            & 0   & 1   &  2 & 2 &  0.4197037681  & 0.4182017399 \\
            & 0   & 1   &  3 & 2 &  0.4201450063  & 0.4176642274 \\
            & 0   & 1   &  1 & 3 &  0.4210954837  & 0.4165892023 \\
        & 0   & 1   &  2 & 3 &  0.4214748454  & 0.4210209862 \\
        &    &      &  1 & 4 &              & 0.4197053158 \\
\hline
\end{array}
$$
\end{center}
$^\star$Best lower bound.
\label{tbl:results-gen-opt}
\end{table}

\section{Exact Computation}
\label{sec:exact-comp}
While it seems difficult to compute the capacity exactly for constraints
such as $\EVEN^{{\otimes}\DD}$ and $\CHARGE(3)^{{\otimes}\DD}$, we can compute
the capacities of constraints in related families:
\begin{theorem}
\label{thm:capac-chg2}
For all positive integers $\DD$,
$$
\capac\left(\CHARGE(2)^{{\otimes}\DD}\right)=\frac{1}{2^{\DD}}.
$$
\end{theorem}
\begin{theorem}
\label{thm:capac-odd}
For all positive integers $\DD$,
$$
\capac\left(\ODD^{{\otimes}\DD}\right)=\frac{1}{2}.
$$
\end{theorem}
\begin{proof}[Proof of Theorem~\ref{thm:capac-chg2}]
Let $S=\CHARGE(2)^{{\otimes}\DD}$. We first show that $\capac(S){\geq}1/2^\DD$.
Let $\Gamma^{(0)},\Gamma^{(1)}$ be the $\DD$-dimensional
arrays of size $2{\times}2{\times}{\ldots}{\times}2$ with entries indexed by $\{0,1\}^\DD$ and given by
$$
\left(\Gamma^{(i)}\right)_{\bldj}=(-1)^{i+\bldj\cdot\bldone}\;\;\;;\;\;\;\bldj\in\left\{0,1\right\}^\DD,
$$
where as usual $\bldone$ denotes the $\DD$-dimensional vector with every entry equal to $1$. Observe that the sum of
every row of both of these arrays is zero.
Now, let $n$ be a positive integer. For any $\DD$-dimensional array of size $n{\times}n{\times}{\ldots}{\times}n$ with entries in $\left\{0,1\right\}$, it can be easily verified that replacing every entry equal to $0$ with $\Gamma^{(0)}$ and every entry equal to $1$ with $\Gamma^{(1)}$ results in a $\DD$-dimensional array of size $2n{\times}2n{\times}{\ldots}{\times}2n$ that satisfies $S$. It follows that  $|S_{2n{\times}2n{\times}{\ldots}{\times}2n}|{\geq}2^{n^{\DD}}$ for all positive integers $n$, which implies $\capac(S){\geq}1/2^\DD$.

We now show that $\capac(S){\leq}1/2^\DD$. For a positive integer $n{\geq}2$, denote by $\N_n^{(0)}$ the set of all even integers in $\left\{0,1,{\ldots},n{-}2\right\}$ and by $\N_n^{(1)}$ the set of all odd integers in $\left\{0,1,{\ldots},n{-}2\right\}$. We shall make use of the following lemma.
\begin{lemma}
\label{lem:chg2-dd=1}
Fix a positive integer $n{\geq}2$, and let $\left(a_i\right)_{i=0}^{n-1}{\subseteq}\left\{+1,-1\right\}$ be a sequence of length $n$. Then $a_0{\ldots}a_{n{-}1}{\in}\CHARGE(2)$ if and only if at least one of the following statements hold.
\begin{enumerate}
\item \label{itm:chg2-dd=1-E} For all $i{\in}\N_{n}^{(0)}$, $a_i{=}{-}a_{i{+}1}$.
\item \label{itm:chg2-dd=1-O} For all $i{\in}\N_{n}^{(1)}$, $a_i{=}{-}a_{i{+}1}$.
\end{enumerate}
\end{lemma}
\begin{proof}
We first show the ``if'' direction. Let $(a_i)_{i=0}^{n{-}1}\subseteq\left\{+1,-1\right\}$ be a sequence for which at least one of statements \ref{itm:chg2-dd=1-E},\ref{itm:chg2-dd=1-O} of the lemma holds. Then clearly, for any integers $0{\leq}i{\leq}j{<}n$, all the terms in the sum $\sum_{k=i}^{j}a_k$, with the possible exception of the first and last terms, cancel. Therefore $|\sum_{k=i}^{j}a_k|\leq|a_i|+|a_j|=2$ and $a_0{\ldots}a_{n{-}1}{\in}\CHARGE(2)$.

As for the ``only if'' direction, let $(a_i)_{i=0}^{n{-}1}{\subseteq}\left\{+1,-1\right\}$ such that $a_0{\ldots}a_{n{-}1}\in\CHARGE(2)$, and consider the presentation of the $\CHARGE$ constraint given in Figure~\ref{fig:chg-presentation} for $b=2$ (and vertices $\{0,1,2\}$). Let $(e_i)_{i=0}^{n{-}1}$ be a path
in this presentation generating $a_0{\ldots}a_{n{-}1}$. It's easily verified that if $\sigma(e_i)=1$, for some $i{\in}[n{-}1]$, then $a_j=-a_{j{+}1}$ for all integers $i{\leq}j{\leq}n{-}2$ such that $j{\equiv}i\mbox{ (mod $2$)}$. Evidently, either $\sigma(e_0)=1$ and so statement~\ref{itm:chg2-dd=1-E} holds, or $\sigma(e_1)=1$ impliying statement~\ref{itm:chg2-dd=1-O}.
\end{proof}

We now return to the claim that $\capac(S){\leq}1/2^\DD$.
Fix a positive integer $n{\geq}2$. For an integer $1{\leq}i{\leq}\DD$, let $\blde^{(i)}{\in}\{0,1\}^\DD$, be the vector, indexed by $\{1,2,{\ldots},\DD\}$, containing $1$ in its $i$th entry and $0$ everywhere else and let $\J_i{\subseteq}[n]^\DD$ denote the subset of all the vectors indexed by $\{1,2,{\ldots},\DD\}$ with a $0$ in the $i$th
entry. For a vector $\bldj{\in}\J_i$, the sequence $\left(\bldj+k\blde^{(i)}\right)_{k=0}^{n{-}1}$, is a sequence of indices of entries of a row in direction $i$ of a $\DD$-dimensional $n{\times}n{\times}{\ldots}{\times}n$ array, and we shall say that it is a sequence in direction $i$. Let $\bldR(n,\DD)$ be the set of all such sequences for all integers $1{\leq}i{\leq}\DD$ and vectors $\bldj{\in}\J_i$, and let $\bldr{\in}\{0,1\}^{|\bldR(n,\DD)|}$ be a binary vector indexed by $\bldR(n,\DD)$. For the purpose of this proof, let us refer to a sequence $\left(a_i\right)_{i=0}^{n{-}1}\subseteq\left\{+1,-1\right\}$ as a phase-$0$ sequence if statement~\ref{itm:chg2-dd=1-E} of Lemma~\ref{lem:chg2-dd=1} holds, and as a phase-$1$ sequence if statement~\ref{itm:chg2-dd=1-O} holds (note that a sequence may be both a phase $0$ and a phase $1$ sequence). Also, we denote by $\A(\bldr){\subseteq}\{+1,-1\}^{*^\DD}$, the set of all $\DD$-dimensional arrays $\Gamma$ of size $n{\times}n{\times}{\ldots}{\times}n$ for which the row $\Gamma_{\bar{\varrho}}$ is a phase-$\bldr_{\bar{\varrho}}$ sequence, for all $\bar{\varrho}=\left(\bldj+k\blde^{(i)}\right)_{k=0}^{n{-}1}{\in}\bldR(n,\DD)$. Then by Lemma~\ref{lem:chg2-dd=1}, we have
\begin{equation}
\label{eq:S_n^d-partitioning}
S_{n{\times}{\ldots}{\times}n}=\bigcup_{\bldr}\A(\bldr).
\end{equation}
We shall give an upper bound on the size of $\A(\bldr)$. For a vector $\bldv{\in}[n]^\DD$, denote by $\rho(i,\bldv)$ the unique sequence in direction $i$ in $\bldR(n,\DD)$ that has $\bldv$ as one of its elements. Let $T_{\bldr,i}:[n]^\DD{\rightarrow}\ZZ^\DD$ be given by:
$$
T_{\bldr,i}(\bldv)=\left\{
\begin{array}{ll}
\bldv+\blde^{(i)} & \mbox{if $v_i\equiv\bldr_{\rho(i,\bldv)}$ (mod $2$)}\\
\bldv-\blde^{(i)} & \mbox{otherwise}
\end{array}
\right.,\, \bldv{\in}[n]^\DD, \bldv=(v_1,{\ldots},v_\DD).
$$
Next, we define the undirected graph $G_\bldr=(V,E_\bldr)$ (without parallel edges), with vertices given by
$$
V=[n]^\DD,
$$
and edges given by
\newcommand{\uedge}{\,\rule[1.7pt]{10pt}{1.00pt}\,}
$$
E_\bldr=\left\{\bldu\uedge\bldv : \mbox{$\bldu,\bldv{\in}V$ and $\bldv=T_{\bldr,i}(\bldu)$ for some integer $1{\leq}i{\leq}\DD$}\right\},
$$
where $\bldu\uedge\bldv$ denotes the undirected edge connecting vertices $\bldu,\bldv$.
It's easy to verify that an array $\Gamma{\in}\{+1,-1\}^{*^\DD}$ of size $n{\times}n{\times}{\ldots}{\times}n$ is in $\A(\bldr)$ iff for every edge $\bldu\uedge\bldv{\in}E_\bldr$, it holds that $\Gamma_\bldu=-\Gamma_\bldv$.
Figure~\ref{fig:G_bldr_example} shows an example of $G_\bldr$ for $\DD=2$.

\begin{figure}[ht]
\begin{center}
\begin{tikzpicture}
\draw[style=help lines,step=0.5cm] (0,0) grid (3,3);
\tikzstyle{vertex}=[circle,fill=black!50!white!50,draw=black,text=black,minimum size=1mm,inner sep=0mm]
\foreach \y/\ytext in {0.25/5,0.75/4,1.25/3,1.75/2,2.25/1,2.75/0}
	\foreach \x/\xtext in {0.25/0,0.75/1,1.25/2,1.75/3,2.25/4,2.75/5}
		\node (v\ytext\xtext) at (\x,\y) [vertex] {};
\foreach \y/\r in {0/1,1/1,2/0,3/1,4/1,5/0} {
	\node at (-0.25, 2.75-\y*0.5) {$\mathit{\y}$};
	\node at (3.25, 2.75-\y*0.5) {$\mathbf{\r}$};
}
\foreach \x/\r in {0/1,1/1,2/1,3/0,4/1,5/0} {
	\node at (\x*0.5+0.25,3.25) {$\mathit{\x}$};
	\node at (\x*0.5+0.25,-0.25) {$\mathbf{\r}$};
}
\tikzstyle{edge}=[line width=0.3mm]
\foreach \y/\r in {0/1,1/1,2/0,3/1,4/1,5/0} {
	\ifnum\r=0 
	   \foreach \xs/\xe in {0/1,2/3,4/5}
	     \draw[style=edge] (v\y\xs) -- (v\y\xe);
	\else
	   \foreach \xs/\xe in {1/2,3/4}
	     \draw[style=edge] (v\y\xs) -- (v\y\xe);
	\fi
}
\foreach \x/\r in {0/1,1/1,2/1,3/0,4/1,5/0} {
	\ifnum\r=0 
	   \foreach \ys/\ye in {0/1,2/3,4/5}
	     \draw[style=edge] (v\ys\x) -- (v\ye\x);
	\else
	   \foreach \ys/\ye in {1/2,3/4} 
	     \draw[style=edge] (v\ys\x) -- (v\ye\x);
	\fi
}
\end{tikzpicture}
\end{center}
\caption{Example of the graph $G_\bldr$ for $\DD=2$, $n=6$. Each entry of $\bldr$ corresponding to a row (column) is written to the right of it (below it). The index of each row (column) is written to its left (above it).}
\label{fig:G_bldr_example}
\end{figure}
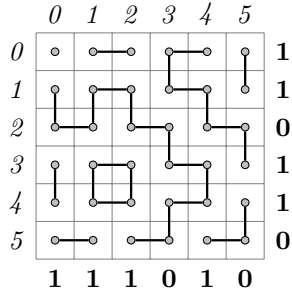

Let $C_1,{\ldots},C_\ell$ be the connected components of $G_\bldr$, and let $\bldv^{(1)},{\ldots},\bldv^{(\ell)}$ be arbitrary vertices such that $\bldv^{(i)}{\in}C_i$, for $i=1,2,{\ldots},\ell$.
It follows that for every vector $\bldb=(b_1,{\ldots},b_\ell){\in}\{+1,-1\}^\ell$, there exists at most one array $\Gamma{\in}\A(\bldr)$ satisfying $\Gamma_{\bldv^{(i)}}=b_i$ for all $i{\in}\{1,2,{\ldots},\ell\}$, and consequently, $|\A(\bldr)|{\leq}2^\ell$ (in fact, while this is not needed for the proof, such an array $\Gamma{\in}\A(\bldr)$ does exist, for any choice of $\bldb$, since each $C_i$ is bipartite; thus $|\A(\bldr)|=2^\ell$).

Now, let $\bldu=(u_1,{\ldots},u_\DD){\in}([n]\setminus\{0,n{-}1\})^\DD$ be a vertex in the ``interior'' of $G_\bldr$. We show that the connected component of $G_\bldr$ containing $\bldu$ has at least $2^\DD$ vertices. To this end, we match each word $\bldw=w_1w_2{\ldots}w_\DD{\in}\{0,1\}^\DD$, with a sequence of vertices  $\left(\pi^{(\bldw,j)}\right)_{j=0}^{\DD}{\subseteq}V$, defined recursively by
$$
\pi^{(\bldw,j)}=\left\{
\begin{array}{ll}
\bldu             & \mbox{if $j$=0}\\
\pi^{(\bldw,j{-}1)}      & \mbox{if $j>0$ and $w_j=0$}\\
T_{\bldr,j}(\pi^{(\bldw,j{-}1)}) & \mbox{if $j>0$ and $w_j=1$}\\
\end{array}
\right..
$$
It's easy to verify that since every $1{\leq}u_i{\leq}n{-}2$, the sequence is well-defined and indeed
$(\pi^{(\bldw,j)})_{j=0}^{\DD}\subseteq [n]^\DD$. Clearly, the sequence is contained entirely in the connected component containing $\bldu$, and so this component contains the vertex $\pi^{(\bldw,\DD)}$.
Write $\pi^{(\bldw,\DD)}{=}(\pi^{(\bldw,\DD)}_1,{\ldots},\pi^{(\bldw,\DD)}_\DD)$. Then for $i{=}1,2,{\ldots},\DD$, it holds that $\pi^{(\bldw,\DD)}_i{=}\bldu_i$ if $w_i{=}0$, and $\pi^{(\bldw,\DD)}_i{=}\bldu_i{\pm}1$ if $w_i{=}1$. Therefore, for two distinct words $\bldw,\bldw'{\in}\{0,1\}^\DD$, the vertices
$\pi^{(\bldw,\DD)}$ and $\pi^{(\bldw',\DD)}$ are distinct as well, and consequently there are $2^\DD$ such vertices.
Thus, the connected component of $G_\bldr$ containing $\bldu$ has at least $2^\DD$ vertices.

It follows that there are at most $n^{\DD}/2^{\DD}$ connected components of $G_\bldr$ containing a vertex in $\{1,2,{\ldots},n{-}2\}^\DD$. There are at most $n^{\DD}-(n{-}2)^\DD$ connected components not containing a vertex in $\{1,2,{\ldots},n{-}2\}^\DD$ and hence the total number of connected components, $\ell$, in $G_\bldr$ satisfies $\ell{\leq}n^{\DD}/2^{\DD}+n^{\DD}-(n{-}2)^\DD$. Hence,
$$
|\A(\bldr)|{\leq}2^{n^{\DD}/2^{\DD}+n^{\DD}-(n{-}2)^\DD}.
$$
Since there are $2^{{\DD}{n^{\DD{-}1}}}$ binary vectors $\bldr{\in}\{0,1\}^{|\bldR(n,\DD)|}$, we obtain from~(\ref{eq:S_n^d-partitioning})
\begin{eqnarray*}
|S_{n{\times}n{\times}{\ldots}{\times}n}|&{\leq}&
\sum_{\bldr}|\A(\bldr)|\\
&{\leq}& 2^{n^{\DD}/2^{\DD}+n^{\DD}-(n{-}2)^\DD+{\DD}{n^{\DD{-}1}}}\\
&=& 2^{n^{\DD}(1/2^{\DD}+1-(1-2/n)^{\DD})+{\DD}n^{\DD{-}1}},\\
\end{eqnarray*}
and the result follows from~(\ref{eq:capacity}).
\end{proof}

\begin{proof}[Proof of Theorem~\ref{thm:capac-odd}]
Let $S$ be the $\DD$-dimensional constraint $\ODD^{{\otimes}\DD}$.
We first show $\capac(S){\geq}{1/2}$. For an integer $n$, let $\B_n\subseteq[2n]^\DD$ be the set of all vectors in $[2n]^\DD$ whose entries sum to an even number, and let $\A_n$ be the set of all binary
$\DD$-dimensional arrays $\Gamma$ of size $2n{\times}2n{\times}{\ldots}{\times}2n$, with entries satisfying
$(\Gamma)_\bldj=0$ for all $\bldj{\in}\B_n$. Then the number of zeros between consecutive `$1$'s, in any row of
an array in $\A_n$ is odd since it must be of the form
$i{-}j{-}1$ for some integers $i,j$--either both odd, or both even. Thus, all such arrays satisfy the constraint $S$,
and since $|\A_n|=2^{(2n)^\DD{-}|\B_n|}=2^{(2n)^\DD/2}$, we have
$|S_{2n{\times}2n{\times}{\ldots}{\times}2n}|\geq2^{(2n)^\DD/2}$ for all positive integers $n$, which implies
$\capac(S){\geq}{1/2}$.

It remains to show that $\capac(S){\leq}{1/2}$. For a positive integer $d$, let $T^{(d)}=\ODD^{{\otimes}d}$. Since for any $d$ positive integers $m_1,{\ldots},m_d$,
$$
\left|(T^{(d)})_{m_1{\times}{\ldots}{\times}m_d}\right|=
 \left|(T^{(d{+}1)})_{m_1{\times}{\ldots}{\times}m_d{\times}1}\right|,
$$
it follows from~(\ref{eq:capac-inf}) that
$\capac(\ODD^{{\otimes}d})$ is non-increasing in $d$. Thus, it's
enough to show $\capac(S){\leq}{1/2}$ for $\DD=1$. Let $n$ be a positive integer. It can be easily verified that any $1$-dimensional array $\Gamma{\in}\ODD_n$ with entries indexed by $[n]$, satisfies either $\Gamma_j=0$ for all even integers $j{\in}[n]$, or
$\Gamma_j=0$ for all odd integers $j{\in}[n]$. It follows that $|\ODD_n|{\leq}2^{\lceil n{/}2\rceil}{+}2^{\lfloor n{/}2 \rfloor}$ which implies the desired inequality
\end{proof}

\end{document}